\RequirePackage{fix-cm}
\documentclass[smallextended]{svjour3}   
\smartqed  
\pdfoutput=1

\usepackage{graphicx}
\usepackage{amsfonts,amssymb,amsmath}
\usepackage{mathptmx}
\usepackage{latexsym}
\usepackage{verbatim,xcolor}
\usepackage{xcolor}
\usepackage{bm}

\newcommand{\kbar}{\overline{k}}
\newcommand{\phibar}{\overline{\phi}}
\newcommand{\xprime}{x^{\prime}}
\newcommand{\vphi}{f_\phi}
\newcommand{\vphib}{\overline{f_\phi}}
\newcommand{\ebar}{\overline{e}}
\newcommand{\HH}{\mathbb{H}}
\newcommand{\FF}{\mathbb{F}}
\newcommand{\RR}{\mathbb{R}}
\newcommand{\Ncal}{\mathcal{N}}
\newcommand{\I}{\mathcal{I}}
\newcommand{\U}{\mathcal{U}}
\newcommand{\Pcal}{\mathcal{P}}
\newcommand{\Hcal}{\mathcal{H}}

\def\K{\mathcal K}
\def\dlhat{\widehat{\delta}}
\def\W{\mathcal W}

\def\khat{\widehat{k}}
\def\ric{\text{Ric}}
\newcommand{\dv}{{\rm d}}

\newcommand{\apx}{a_{\perp, x}}
\newcommand{\apy}{a_{\perp, y}}
\newcommand{\Hap}{\mathcal H_{\text{app}}}
\newcommand{\tHap}{\widetilde{\mathcal H}_{\text{app}}}

\begin{document}

\title{Many-body excitations in trapped Bose gas: A non-Hermitian view 
\thanks{The work of the second and third authors (DM and SS) was partially supported by Grant No.~1517162 of the Division of Mathematical Sciences (DMS) of the National Science Foundation (NSF).}
}

\date{\today}

\titlerunning{Many-body excitations in trapped Bose gas: A non-Hermitian view}        

\author{Manoussos G. Grillakis \and Dionisios Margetis  \and    Stephen Sorokanich}

\authorrunning{M.~G. Grillakis, D. Margetis, S. Sorokanich} 

\institute{M.~G. Grillakis \at 
Department of Mathematics, University of Maryland, College Park, MD 20742, USA\\
\email{mggrlk@umd.edu}
\and 
D. Margetis \at
               Institute for Physical Science and Technology, and Department of Mathematics, and Center for Scientific Computation and Mathematical Modeling, University of Maryland, College Park, MD 20742, USA\\
              \email{diom@umd.edu}           
           \and
           S. Sorokanich \at
              Department of Mathematics, University of Maryland, College Park, MD 20742, USA\\
              \email{ssorokan@umd.edu}
}


\maketitle

\begin{abstract}
We provide the analysis of a physically motivated model for a trapped dilute Bose gas with repulsive pairwise atomic interactions at zero temperature. Our goal is to  describe aspects of the excited many-body quantum states by accounting for the scattering of atoms in pairs from the macroscopic state (condensate). We formally construct a many-body Hamiltonian, $\Hap$, that is quadratic in the Boson field operators for noncondensate atoms. This $\Hap$ conserves the total number of atoms. Inspired by Wu ({\em J.\ Math.\ Phys.}, 2:105--123, 1961), we apply a \emph{non-unitary} transformation to $\Hap$. Key in this non-Hermitian view is the pair-excitation kernel,  which in operator form obeys a Riccati equation. In the stationary case, we develop an existence theory for solutions to this operator equation by a variational approach. We connect this theory to the one-particle excitation wave functions heuristically derived by Fetter ({\it Ann.\ Phys.}, 70:67--101, 1972). These functions solve an eigenvalue problem for a $J$-self-adjoint operator. From the non-Hermitian Hamiltonian, we derive a one-particle nonlocal equation for low-lying excitations, describe its solutions, and recover Fetter's excitation spectrum. Our approach leads to a description of the excited eigenstates of the reduced Hamiltonian in the $N$-particle sector of Fock space.



\end{abstract}

\section{Introduction}
\label{sec:Intro}

In Bose-Einstein condensation (BEC) integer-spin particles (Bosons) occupy en masse a single-particle macroscopic quantum state, known as the `condensate', at extremely low temperatures. The first experimental realizations of BEC in trapped atomic gases in 1995 \cite{Anderson1995,Ketterle1995} -- nearly 80 years after its first prediction by Bose and Einstein -- were the subject of the 2001 Nobel Prize in Physics~\cite{Cornell-rev2002,Ketterle-rev2002}. Since 1995, the experimental and theoretical research in harnessing ultracold atomic gases has grown considerably. An emergent and far-reaching advance in applied physics is the highly precise manipulation of atoms by optical or magnetic  means~\cite{Chinetal2010,Cooperetal2019,Dalfovoetal1999,Fetter2009,Tomzaetal2019}. 

The dilute atomic gas is amenable to a systematic analysis mainly because of the length scale separation inherent to this system~\cite{LiebSeiringer-book}. The following length scales are involved in this problem: (i) The low-energy scattering length, $a$, which expresses the strength of the atomic interactions and is positive for repulsively interacting atoms. (ii) The mean interatomic distance, $\ell$, which is set by the mean density of the gas. (iii) The de Broglie wavelength, $\ell_{\text{dB}}$, of each  atom. For many experimental situations, it is reasonable to assume that $a\ll \ell\ll \ell_{\text{dB}}$. If a trapping potential is applied externally, another length scale is the linear size of the trap, which can be of the same order as or larger than $\ell_{\text{dB}}$. The gas diluteness usually amounts to the condition $a\ll \ell$, and a macroscopic quantum state may exist if $\ell \lesssim \ell_{\text{dB}}$. 

The realization of BEC in atomic gases has sparked various investigations in the modeling and analysis of nonlinear dynamics and out-of-equilibrium phenomena in Boson systems~\cite{Chinetal2010,Stamper-Kurn2013,Nam2017-chapter,Bossmannetal2020,Morsch2006,PethickSmith2008,Margetis2012}. 
In this context, the usual mean field approach~\cite{Schlein2017-chapter}, which exclusively relies on the macroscopic wave function for the condensate, is often (but not always) employed. Despite the success of this approach for many phenomena, its limitations have been early recognized~\cite{wu61}. In applications, this mean field theory cannot capture, for example, the condensate depletion; see, e.g.,~\cite{Xuetal2006}. The modeling of such an effect requires a systematic description of  truly many-body dynamics, in particular  \emph{pair excitations}~\cite{Bogoliubov1947,leehuangyang,wu61,wu98,Seiringer,MGM}.  

In this paper, our goal is to describe aspects of the excited many-body eigenstates of an interacting Bose system in an external trapping potential. To this end, we employ a simplified effective model: a quadratic-in-Boson-field-operators Hamiltonian, called $\Hap$, that captures pair creation.\footnote{Our terminology in this paper differs from that in~\cite{Zagrebnov2001} where such Hamiltonians are characterized as \emph{bilinear} in creation and annihilation Boson operators associated with noncondensate particles.} This $\Hap$ commutes with the particle number operator; thus, the total number of particles is conserved in Fock space. We formally construct $\Hap$ from the full many-body Hamiltonian with a regularized interaction potential. By invoking the formalism of Wu~\cite{wu61,wu98}, we apply a non-unitary transformation to $\Hap$. For stationary states, we analyze the role of the \emph{pair excitation kernel}, $k$, a function of two spatial variables introduced by this transformation. This $k$ expresses the scattering of atoms from the condensate in pairs; and satisfies a nonlinear integro-differential equation. We develop an \emph{existence theory} for this equation by a variational approach. Our treatment reveals a previously unnoticed connection of $k$ to the one-particle excitation wave functions, $u_j$ and $v_j$, introduced independently by Fetter~\cite{fetter72}. These functions obey a system of linear partial differential equations (PDEs). Our analysis sheds light on the existence of the eigenfunctions $u_j$ and $v_j$, and eigenvalues $E_j$, for this system. By the non-Hermitian Hamiltonian that results from the transformed $\Hap$, we derive a nonlocal PDE for phonon-like excitations in the trap; and express its solutions in terms of $u_j$ and $v_j$. We rigorously relate the eigenvalues of this nonlocal PDE with $E_j$; and recover the excitation spectrum obtained in~\cite{fetter72}. Our approach yields an explicit construction of the excited many-body eigenstates of $\Hap$ in the appropriate sector of Fock space.

Our specific tasks and results can be outlined as follows (see also Sect.~\ref{sec:results}):

\begin{itemize}

\item Starting from the full many-body Hamiltonian with positive and smooth interaction and trapping potentials, we formally apply an approximation scheme that leads to a Hamiltonian, $\Hap$, quadratic in the Boson field operator for noncondensate particles. This $\Hap$ is a regularized version of the model developed by Wu~\cite{wu98}. The total number of particles is conserved. 


\item For stationary states, we invoke ideas of pair excitation~\cite{wu61}. A key ingredient is the pair excitation kernel, $k$, which is involved in a \emph{non-unitary} exponential transformation of $\Hap$. In operator form this $k$ satisfies a Riccati equation.

\item 
By constructing a functional of $k$, we prove the existence of solutions to the operator Riccati equation in an appropriate space. Our analysis, based on a variational principle, differs from many previous treatments of the operator Riccati equation. We indicate the possibility of multiple solutions for $k$, and distinguish the physically relevant, unique solution via a restriction on the operator norm of $k$. 

\item We provide an explicit construction of the many-body eigenstates of $\Hap$ in the $N$-particle sector of Fock space. We also show that the spectrum of $\Hap$ is positive and discrete.

\item We show that the existence of solutions to the equation for $k$ implies the existence of solutions to the eigenvalue problem for the one-particle excitation wave functions $u_j$ and $v_j$ with a regularized interaction in~\cite{fetter72}. We employ aspects of the theory of $J$-self-adjoint operators by Albeverio and collaborators~\cite{AlbeverioMotovilov2019,Albeverio2009,AlbeverioMotovilov2010,Tretter2016}. Hence, we connect the apparently disparate approaches for low-lying (phonon-like) excitations of Bosons in a trap by Fetter~\cite{fetter72} and Wu~\cite{wu61,wu98}. 

\item As a consequence of the non-unitarily transformed $\Hap$, we formally derive a one-particle PDE (``phonon PDE'') for single-particle excitations in the Bose gas.  
By restricting the operator norm of $k$, we show that the point spectrum of the Schr\"odinger operator of the phonon PDE coincides with physically admissible eigenvalues $E_j$ of the PDEs for $(u_j, v_j)$, in agreement with~\cite{fetter72}.

\end{itemize}

At the risk of redundancy, we repeat that our work reveals a nontrivial connection between the non-Hermitian approach of Wu~\cite{wu61} to the Hermitian framework of Fetter~\cite{fetter72} for low-lying excitations via the operator theory of Albeverio and collaborators~\cite{AlbeverioMotovilov2019,Albeverio2009,AlbeverioMotovilov2010,Tretter2016}; see Fig.~\ref{fig:Schematic}. Because of this connection, we can show the solvability of Fetter's eigenvalue problem with a regularized interaction potential~\cite{fetter72}.

\begin{figure}[h]
\includegraphics[scale=0.35,trim=0in 2.5in 0in 0.55in]{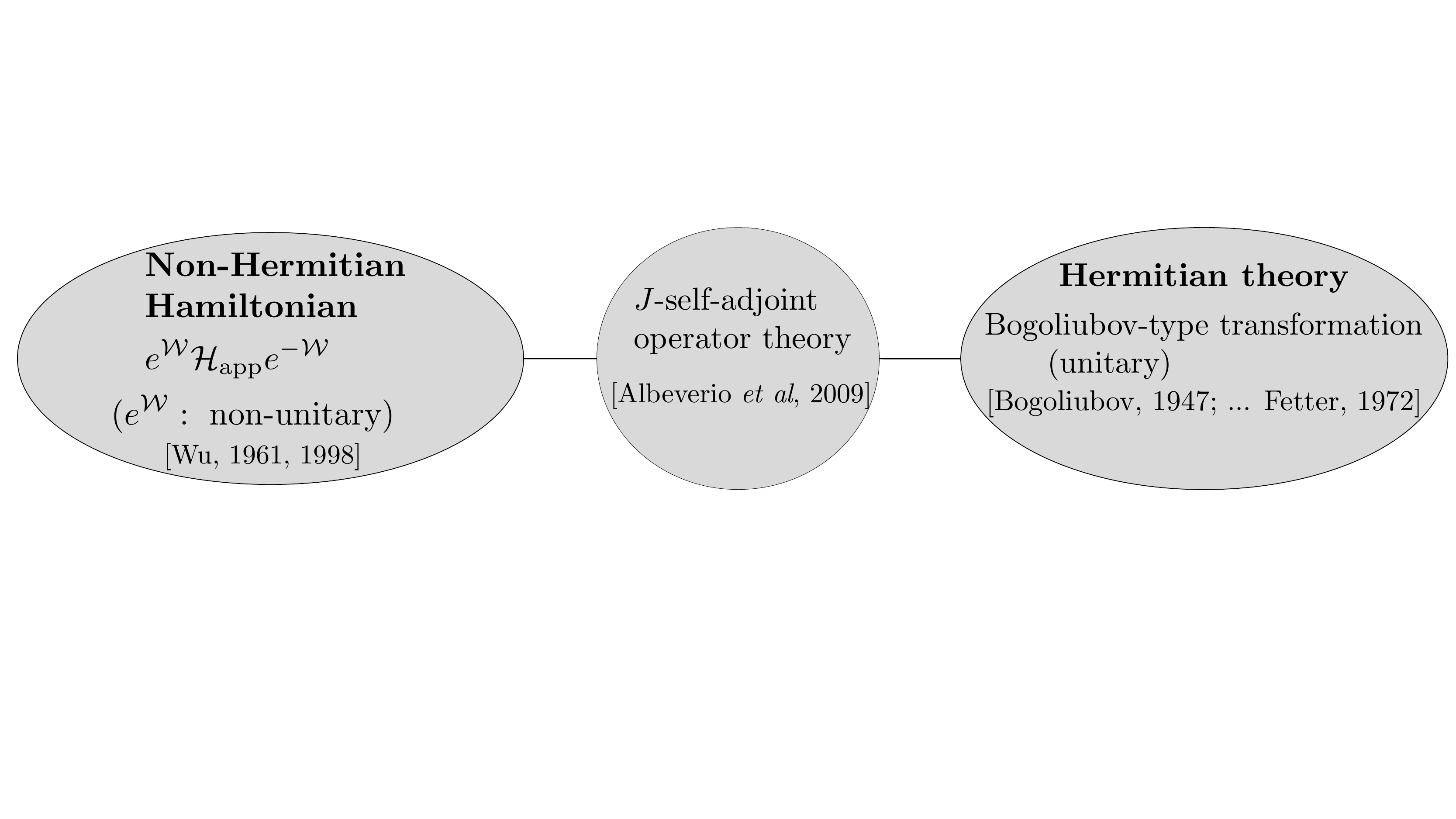}
\centering{}\caption{Schematic for the connection of two main physical approaches (left and right panels) to the problem of excitations in the Bose gas via abstract operator theory (central panel).}
\label{fig:Schematic}
\end{figure}

Our main focus is on the analysis of low-dimensional PDEs that formally result from a non-unitary transformation of the quadratic many-body Hamiltonian $\Hap$~\cite{wu61,wu98}. This Hamiltonian is a starting point of our analysis. Notably, $\Hap$ is physically motivated and is derived heuristically from the full many-particle Hamiltonian, as we show by using a regularized interaction potential. In our procedure, we fix the (conserved) total number of atoms at the value $N$ ($N\gg 1$). A rigorous justification for $\Hap$ lies beyond our scope.  In a similar vein, we sketch plausibility arguments for the extraction of low-dimensional PDEs such as the equation for $k$. On the other hand, the analysis of solutions to these equations is rigorous.  The thermodynamic limit ($N\to\infty$) is not treated by our theory. For aspects of this limit, see, e.g.,~\cite{Cederbaum2017,Lewin2014}.


Notably, the non-unitarily transformed quadratic Hamiltonian considered here has space-time reflection symmetry. This model suggests an example of a physical non-Hermitian quantum theory (see~\cite{Bender2007}). The systematic comparison of the non-Hermitian framework involving the pair-excitation kernel to concepts emerging from the theory of space-time-reflection-symmetric Hamiltonians is left for future work.

 The motivation for the non-Hermitian view of our paper is outlined in Sect.~\ref{subsec:motivation}. Previous related works are discussed in Sect.~\ref{subsec:prev_works}. The underlying mathematical formalism including Fock space concepts is reviewed in Sect.~\ref{subsec:notation}. The paper organization is sketched in Sect.~\ref{subsec:organiz}. (The reader who wishes to skip the remaining introduction and read highlights of our results is deferred to Sect.~\ref{sec:results}.)

\subsection{Motivation: Why a non-Hermitian view?}
\label{subsec:motivation}

The reader may raise the following question: What is our motivation for focusing on the non-Hermitian approach of~\cite{wu61,wu98}? After all, non-unitary transformations are often deemed as mathematically hard to deal with. Our motivation is twofold. 

First, from a physics perspective, it can be argued that the formalism involving the pair-excitation kernel is a natural extension of the systematic treatment by Lee, Huang and Yang for the setting with translation invariance and periodic boundary conditions~\cite{leehuangyang}. In their case, the eigenvectors of the many-body Hamiltonian can be approximately expressed in terms of the action of a non-unitary operator, $e^{\K}$, on finite superpositions of tensor products of one-particle momentum ($\bm p$) states~\cite{leehuangyang}. The exponent $\K$ is of the form~\cite{leehuangyang,wu61} 
\begin{equation*}
	\K=-\frac{1}{2}\sum_{\bm p\neq 0}\alpha(\bm p)\,a_{\bm p}^* a_{-\bm p}^*~,
\end{equation*}
where $a_{\bm p}$ ($a_{\bm p}^*$) is the annihilation (creation) operator at one-particle momentum $\bm p$, $\bm p \in (2\pi/L) \mathbb{Z}^3$ and $L$ is the linear size of the periodic box. The function $\alpha(\bm p)$, where $\alpha: (2\pi/L)\mathbb{Z}^3\rightarrow \RR_+$, yields the phonon spectrum. In the above, the operator $a_0$ was replaced by $\sqrt{N}$ times the identity operator, which amounts to the Bogoliubov approximation in the periodic setting~\cite{Seiringer}. Each term in the series for $\K$ describes the excitation of particles from the condensate to pairs of opposite momenta. 

Adopting Wu's extension of the above treatment to non-translation invariant settings~\cite{wu61,wu98}, in the case of stationary states we consider a Hamiltonian that conserves the total number of atoms. In addition, we replace the exponent $\K$ by an integral, $\W$, over $\mathbb{R}^3\times \mathbb{R}^3$. This integral involves the pair excitation kernel, $k$, a symmetric function of two spatial variables  ($k: \RR^3\times \RR^3\rightarrow \mathbb{C}$), viz.,
\begin{equation*}
\W=-(2N)^{-1}\iint_{\RR^6} \dv x\, \dv y\ a_x^\ast a_y^\ast\, k(x,y)\,a(\overline \phi)^2.
\end{equation*}
Here, $a_x^\ast$ is the Boson field creation operator at position $x$, $\phi$ denotes the condensate wave function ($\phi: \RR^3\rightarrow \mathbb{C}$), $a(\bar\phi)$ is the Boson field annihilation operator for the single-particle state $\phi$, and $\phi$ is assumed to be orthogonal to $k$; see Sect.~\ref{subsec:notation}. The kernel $k$ is found to obey a nonlocal and nonlinear PDE~\cite{wu61,wu98}. In this paper, we rigorously study the existence of stationary solutions to this PDE and explore possible implications. Note that $a(\overline\phi)$ is not treated as a $c$-number within our approach.

We show that the formalism based on the operator $\W$ yields an excitation energy spectrum for the Bose gas in agreement with the one derived by Fetter~\cite{fetter72}. His  approach  invokes a Bogoliubov-type rotation of Boson field operators in the space orthogonal to $\phi$ which keeps intact the Hermiticity of the many-body Hamiltonian. Fetter's Hamiltonian does not commute with the particle number operator. Here, we place emphasis on the pair-excitation kernel, $k$, in the context of a non-Hermitian Hamiltonian that conserves the total number of particles. We \emph{explicitly construct the excited many-body eigenstates in the  appropriate ($N$-particle) sector of Fock space}. 

Another reason for our choice of the non-Hermitian view is that this illustrates, and exploits, previously unnoticed connections of abstract operator theory to phonon-like excitations in Boson dynamics. In our analysis we identify the governing equation for $k$ with an operator Riccati equation.  The latter has been studied extensively by Albeverio and coworkers in an abstract context~\cite{AlbeverioMotovilov2019,Albeverio2009,AlbeverioMotovilov2010}; see also~\cite{Tretter2016,Kostrykin03-chap,Tretter-book}. Our existence theory for this equation, based on a variational approach, seems to differ from existence proofs found in these works. Since we focus on $k$, our formalism has a different flavor from the variational approaches for operator matrices in~\cite{Tretter-book}. The Riccati equation for $k$ here is inherent to the non-Hermitian formalism for the Boson system~\cite{wu61}. We rigorously establish that: The excitation spectrum by Fetter's approach~\cite{fetter72} comes from the eigenvalue problem for a $J$-self-adjoint operator intimately connected to the equation for $k$ from Wu's treatment~\cite{wu61,wu98}. In our analysis, we use a regularized interaction potential in the place of the delta-function potential used in~\cite{fetter72,wu61,wu98}. Our findings for the excitation spectrum are independent of the particle-conserving (or not) character of the quadratic Hamiltonian. In contrast, the construction of the many-body eigenstates relies on the particle number conservation.

\subsection{On related past works}
\label{subsec:prev_works} 

The quantum dynamics of the Bose gas has been the subject of numerous studies. It is impossible to exhaustively list this bibliography. Here, we make an attempt to place our work in the appropriate context of the existing literature. For a broad view on Boson dynamics, the interested reader may consult, e.g.,~\cite{LiebSeiringer-book,Seiringer,Zagrebnov2001,Lewin2016,Schlein2017-chapter,MGM,Dalfovoetal1999}.

Mean field limits of Boson dynamics are usually captured by nonlinear Schr\"odinger-type equations for the condensate wave function~\cite{Gross61,Pitaevskii61,wu61}. Such limits have been rigorously derived from kinetic hierarchies in distinct scaling regimes for the atomic interactions in the works by Erd\H{o}s, Schlein, Yau and collaborators; for the Gross-Pitaevskii regime, see~\cite{ErdosSchleinYau2010}. Our focus in this paper is different. We primarily address the analysis of low-order PDEs that aim to provide corrections to the mean field dynamics of a given quadratic Hamiltonian, and their relation to the excitation spectrum. 

Second-order corrections to the mean field time evolution have been studied rigorously through a Bogoliubov-type transformation in~\cite{GMM2010,GMM2011}. Although that work is inspired by Wu's approach~\cite{wu61,wu98}, it is not strictly faithful to his formalism.  In fact, in~\cite{GMM2010,GMM2011} the many-body Hamiltonian is transformed \emph{unitarily}  whereas in~\cite{wu61,wu98} the corresponding transformation is \emph{non-unitary}. In this paper, we take a firm step towards exploring aspects of the latter approach via a minimal model, by applying a non-unitary transformation to an effective quadratic Hamiltonian in the stationary setting. By this model, we describe the excited many-body quantum states of the gas.

Wu's formal treatment of the interacting Bose system in non-translation invariant settings aims to transcend the mean field limit~\cite{wu61,wu98}. In mathematics, this approach has motivated the use of the pair-excitation kernel $k$ as a means of improving error estimates for the time evolution of Bosons~\cite{MGM}. It has been shown that a unitary, Bogoliubov-type transformation of the many-body Hamiltonian that involves $k$ yields considerably improved Fock space estimates~\cite{GM2013-b,GM2013-a,GM2017,GMM2010,GMM2011}. A price to pay is that $k$ satisfies a nonlocal evolution PDE coupled with the Gross-Pitaevskii equation for the condensate wave function. Because of the use of a unitary transformation in~\cite{GM2013-b,GM2013-a,GM2017,GMM2010,GMM2011,MGM}, the PDE for $k$ in those works is different from the one in~\cite{wu61,wu98}.

There are many other papers that tackle the problems of quantum fluctuations around the mean field limit and the excitation spectrum of the Bose gas in the mathematics literature; see, e.g.,~\cite{Boccatoetal2020,Boccato2020,BrenneckeSchlein2019,Nam2017-chapter,NamNapiorkowski2017,NamNapiorkowski2017-II,NamSolovej2016,NamSeiringer2015,Lewin2015,Lewin2014,Derezinski2014,Seiringer2011,Cornean2009}. A comprehensive review of some of the challenges in analyzing many-body excitations in the periodic box is given by Seiringer~\cite{Seiringer}. Central roles in many treatments of the excitation spectrum are played by the Bogoliubov approximation and the Bogoliubov (unitary) transformation. In particular, in~\cite{Lewin2014} Lewin and coworkers tackle aspects of this problem by use of a quadratic Hamiltonian with a trapping potential via Fock space techniques in the limit $N\to \infty$. \color{black} In these works, the dominant view is Hermitian.  


In the physics literature, the excitation spectrum of the Bose gas in non-translation invariant settings has been explicitly described by many authors; for reviews see, e.g.,~\cite{Dalfovoetal1999,Gardiner1997,Griffin1996,Leggett2001,Ozeri2005,Rovenchak2016}. We single out the work by Fetter~\cite{fetter72,fetter96,Fetter2009} who formally addresses this problem through an intriguing linear PDE system. The existence of solutions to this system has not been studied until now.  The underlying many-body formalism relies on a unitary, Bogoliubov-type transformation of Boson field operators for noncondensate particles. This leads to a formula for the excitation spectrum in terms of the eigenvalues, $E_j$, of the PDE system~\cite{fetter72}. This formalism has been invoked in the modeling of phonon scattering~\cite{Danshita2006} and condensate fluctuations at finite temperatures~\cite{Griffin1996}.

In a nutshell, our analysis brings forth an intimate mathematical connection of Fetter's theory~\cite{fetter72} to Wu's approach~\cite{wu61} (see Fig.~\ref{fig:Schematic}). 
Regarding the existence theory for the operator Riccati equation obeyed by $k$,  we develop a variational approach which significantly differs from the previously invoked fixed-point argument~\cite{Albeverio2009,Tretter2016}. We show that this theory naturally implies the existence of solutions to Fetter's PDE system for a regularized interaction potential~\cite{fetter72}. We also construct the eigenvectors of our quadratic Hamiltonian in the $N$-sector of Fock space. The particle-conserving character of the Hamiltonian and the use of kernel $k$ are key in this construction.

\subsection{Notation and terminology}
\label{subsec:notation}

\begin{itemize}

\item[$\bullet$] 

The symbol $\overline{f}$ denotes the complex conjugate of $f$, while $A^\ast$ stands for the Hermitian adjoint
of operator $A$. Also, the symbol $\overline{A}$ indicates the operator which acts according to $\overline{A}[f]=\overline{\{A[f]\}}$ for all functions $f$ in the domain of $A$. 

\item[$\bullet$]

In the symbol $\int$, the integration limits are omitted. The corresponding region is $\mathbb{R}^3$ (for $\int \dv x$) or $\mathbb{R}^3\times \mathbb{R}^3$ (for $\int \dv x\,\dv y$).

\item[$\bullet$] 

The (symmetric) inner product of complex-valued $f,g\in L^2(\mathbb{R}^3)$ is defined by
\begin{equation*}
\langle \overline{f},g\rangle=\int{\dv x\ \big\{ \overline{f(x)}\,g(x)\big\}}~.
\end{equation*}
The respective inner product of complex-valued $f,g\in L_V^2(\mathbb{R}^3)$ is $\langle \overline{f},Vg\rangle$ for positive external potential $V(x)$. The $L^2$-norm of $f$ is denoted $\Vert f\Vert_2$.  
For some operator $k$, the (symmetric) inner product of $f(x)$ and $k(x,g)$ is denoted
$\big\langle f,k(\cdot,g)\big\rangle$. 

\item[$\bullet$]

Function spaces on $\mathbb{R}^{d}$ (e.g., $d=3$) are denoted by lowercase gothic letters, viz., 
\begin{equation*}
    \mathfrak{h}(\mathbb{R}^d) := L^2(\mathbb{R}^d)~,\quad  \mathfrak{h}^1(\mathbb{R}^d) := H^1(\mathbb{R}^d)~, \quad \mathfrak{h}_V^1(\mathbb{R}^d) := H^1(\mathbb{R}^d)\cap L_V^2(\mathbb{R}^d)~.
\end{equation*}
We write $\mathfrak{h}$, $\mathfrak{h}^1$, $\mathfrak{h}_V^1$ for these spaces if $d=3$. As an exception to this notation, we define  $\phi^{\perp}:=
 \big\{e\in \mathfrak{h}^{1}_{V}\ \big\vert\ e\perp\phi\big\}$ where $\phi\in \mathfrak{h}^1$ is the condensate wave function. 

\item[$\bullet$]

For a given ordered set $\{e_j(x)\}_j\subset \mathfrak{h}$, we occasionally use the symbol $\langle A\rangle _j$ for  the inner product $\langle e_j, a(\cdot,e_j) \rangle$, taking $A:=a(x,y)$.

\item[$\bullet$]

The symbol $(\upsilon\ast g)(x)$ denotes the convolution integral $\int \dv y\,\upsilon(x-y) g(y)$.


\item[$\bullet$]

 The space of bounded linear operators on $\mathfrak{h}$ is denoted $\mathfrak{B}(\mathfrak{h})$, with norm $\|\cdot\|_{\mathrm{op}}$. Also, the space of trace-class operators on $\mathfrak{h}$ is denoted $\mathfrak{B}_1(\mathfrak{h})$ with norm
\begin{equation*}
\|A\|_{\mathfrak{B}_1(\mathfrak{h})}=\|A\|_1= \mathrm{tr}|A|~,\quad \forall\,\,A\in\mathfrak{B}_1(\mathfrak{h})~.
\end{equation*}
Similarly, the space of Hilbert-Schmidt operators on $\mathfrak{h}$ is $\mathfrak{B}_2(\mathfrak{h})$ with norm
\begin{equation*}
\|A\|_{\mathfrak{B}_2(\mathfrak{h})}=\|A\|_2 = (\mathrm{tr}| A^\ast A|)^{1/2}~,\quad \forall\,\,A\in\mathfrak{B}_2(\mathfrak{h})~.
\end{equation*}
The space of compact operators on $\mathfrak{h}$ is $\mathfrak{B}_0(\mathfrak{h})$. 
Note the  inequalities
\begin{equation*}
\|A\|_{\mathrm{op}} \le \|A\|_{2} \le \|A\|_{1}~, 
\end{equation*}
and the inclusions
$\mathfrak{B}_1(\mathfrak{h})\subseteq\mathfrak{B}_2(\mathfrak{h})\subseteq\mathfrak{B}_0(\mathfrak{h})\subseteq\mathfrak{B}(\mathfrak{h})$.

\item[$\bullet$]
We express operators on $\mathfrak{h}$ by use of their integral kernels which we denote by lowercase greek or roman letters. For example, we employ the expression $\delta(x,y)$, in place of $\delta(x-y)$, of  the Dirac mass for the identity operator. In this vein, an effective one-particle Hamiltonian of interest is denoted by the singular kernel   
\begin{equation*}
h(x,y):= \big\{-\Delta+V(x)+N(\upsilon\ast|\phi|^2)(x)\big\}\delta(x,y) + N\phi(x)\upsilon(x-y)\overline{\phi(y)}-\mu~,
\end{equation*}
where $\phi(x)$ is the condensate wave function, $V(x)$ is the trapping potential, $\upsilon(x)$ is the two-body interaction potential, and $\mu$ is a constant. 
Another example of notation is $k(x,y)$ for the pair-excitation operator. We use the superscript `$T$' for a kernel to denote its transpose. The star ($\ast$) as a superscript indicates the adjoint (complex conjugate and transpose) kernel; e.g., $k^{\ast}(x,y)=\overline{k(y,x)}$. We write $k\in\mathfrak{S}$ to mean that `the operator with integral kernel $k$' belongs to the space $\mathfrak S$, e.g., for $\mathfrak S=\mathfrak B_2(\mathfrak{h})$.

\item[$\bullet$]

The composition of operators $h$ and $k$ is expressed by 
\begin{equation*}
(h\circ k)(x,y) := \int{\dv x'\, \big\{h(x,x')k(x',y)\big\}}~.
\end{equation*}

\item[$\bullet$]

If a bounded operator $k\in\mathfrak{B}(\mathfrak{h})$ acts on $f\in\mathfrak{h}$, the result is the function 
\begin{equation*}
k(x\,,f) := \int{\dv x'\,\{k(x,x')f(x')\}}~,\quad \mbox{or}
\quad k(f,x'):=\int{\dv x\,\{f(x)k(x,x')\}}~.
\end{equation*}
The same notation is used for kernels corresponding to unbounded operators, with the understanding that the domain of such an operator is defined appropriately.

\item[$\bullet$] 

For $f, g \in \mathfrak{h}$
the tensor-product operator corresponding to integral kernel $f(x)\overline{g(x')}$ is sometimes expressed as $f\otimes g$. The symmetrized tensor product of $f,g$ is 
\begin{equation*}
f\otimes_{\mathrm{s}}g:= \frac{1}{\sqrt{2}}\big\{f\otimes g +g\otimes f\big\}~.
\end{equation*}

\item[$\bullet$] For the condensate wave function $\phi\in\mathfrak{h}$ with $L^2$-norm $\|\phi\|_2=1$,  the projection operator $\widehat\delta:\mathfrak{h}\to\mathfrak{h}$ is defined by 
\begin{equation*}
\widehat{\delta}(x,y)=\delta(x,y)-\phi(x)\overline{\phi(y)}~.
\end{equation*}

\item[$\bullet$]

The Bosonic Fock space $\FF$ is a direct sum of $n$-particle symmetric $L^2$-spaces, viz.,
\begin{equation*}
\FF=\bigoplus_{n=0}^{\infty}\FF_{n}~;\quad \FF_0=\mathbb{C}~,\quad \FF_n=L^2_s(\mathbb{R}^{3n})\ \mbox{if}\ n\ge 1~. 
\end{equation*}
Hence, vectors in $\FF$ are described as sequences $\{u^n\}$ of $n$-particle wave functions where $u^n\in L^2_s(\mathbb{R}^{3n})$, $n\ge0$. The inner product of $\vert u\rangle=\{u^n\}, \vert w\rangle=\{w^n\}\in\FF$ is 
\begin{equation*}
\langle u,w\rangle_{\FF}:=\sum_{n=0}^\infty{\langle \overline{u}^n,w^n\rangle_{L^2(\mathbb{R}^{3n})}}~,
\end{equation*}
which induces the norm $\Vert |u\rangle\Vert=\sqrt{\langle u,u\rangle_{\FF}}$.  
We employ the bra-ket notation for Schr\"odinger state vectors in $\FF$ to distinguish them from wave functions in $L^2_s(\mathbb{R}^{3n})$. We often write the inner product of $\vert u\rangle$ with $\mathcal A\vert w\rangle$ ($\mathcal A: \FF \mapsto \FF$) as $\langle u\vert \mathcal A \vert w\rangle$. The vacuum state in $\FF$ is $\vert  vac\rangle:={\{1,0,0,\dots\}}$, where the unity is placed in the zeroth slot. A symmetric $N$-particle wave function, $\psi_N\in L^2_s(\mathbb{R}^{3N})$, has a natural embedding into $\FF$ given by $\vert\psi\rangle_N=\{0,0,\dots,\psi_N(x),0,\dots\}$, where $\psi_N(x)$ is in the $N$-th slot. The set of such state vectors $\vert \psi\rangle_N$ is the `$N$-th fiber' ($N$-particle sector) of $\FF$, denoted $\mathbb{F}_N$. We sometimes omit the subscript `$N$' in $\vert\psi\rangle_N$, simply writing $\vert\psi\rangle$.

\item[$\bullet$]

A Hamiltonian on $L_s^2(\mathbb{R}^{3N})$ admits an extension to an operator on $\FF$. This extension is carried out via the Bosonic field operator $a_x$ and its adjoint, $a_x^\ast$, which are indexed by the spatial coordinate $x\in\mathbb{R}^3$. To define these field operators, first consider the annihilation and creation operators for a one-particle state $f\in \mathfrak{h}$, denoted by $a(\overline{f})$ and $a^\ast(f)$. These operators act on $\vert u\rangle=\{u^n\}\in\FF$ according to 
\begin{equation*}
\big(a(\overline{f})\vert u\rangle\big)^n:=\sqrt{n+1}\int{\dv x\, \overline{f(x)}\,u^{n+1}(x,x_2,\dots,x_n)}~,
\end{equation*}
\begin{equation*}
\big(a^\ast(f)\vert u\rangle\big)^n:=\frac{1}{\sqrt{n}}\sum_{j\le n}{f(x_j)\,u^{n-1}(x_1,\dots,x_{j-1},x_{j+1},\dots,x_n)}~.
\end{equation*}
We often use the symbols $a_{\overline f}:=a(\overline f)$ and $a^*_f:= a^*(f)$. Also, given an orthonormal basis, $\{e_j(x)\}_j\subset \mathfrak{h}$, we will write $a_j^\ast$ in place of $a^\ast (e_j)$ and $a_j$ in place of $a(\overline{e_j})$.  


\item[$\bullet$]

The Boson field operators $a_x^\ast,\, a_x$ are now implicitly defined via the integrals
\begin{equation*}
a^\ast_f =\int{\dv x\, \big\{f(x)\,a_x^\ast \big\}}~,\quad 
a_{\overline f} = \int{\dv x \,\big\{\overline{f(x)}\,a_x\big\}}~.
\end{equation*}
By the orthonormal basis $\{e_j(x)\}_j$, the field operators are expressed by
\begin{equation*}
a^\ast_x =\sum_j{e_j(x)\,a^\ast_j}~,\quad 
a_x = \sum_j{\overline{e_j(x)}\,a_j}~.
\end{equation*}
The canonical commutation relations 
$[a_x,a_y^\ast] =\delta(x-y)$, $[a_x,a_y] = 0$ then follow.

\item[$\bullet$] The Boson field operators orthogonal to the condensate $\phi\in\mathfrak{h}$ are defined by
%
\begin{equation*}
a_{\perp,x}=\int\dv y\,\left\{\dlhat(x,y)a_{y}\right\}
=\int \dv y\,\left\{a_{y}\dlhat^{T}(y,x)\right\}~,
\end{equation*}
\begin{equation*}
a^{\ast}_{\perp,x}=\int
\dv y\,\left\{\dlhat^{T}(x,y)a^{\ast}_{y}\right\}
=\int
\dv y\,\left\{a^{\ast}_{y}\dlhat(y,x)\right\}~.
\end{equation*}
We can decompose the Boson field operators according to the equations
\begin{equation}\label{eq:ax-dec}
a_x=\apx + \phi(x) a_{\overline{\phi}}~,\quad 
a^\ast_x=\apx^\ast + \overline{\phi(x)} a^\ast_\phi~.
\end{equation}
It is worthwhile to notice the commutation relations
\begin{equation*}
\big[a_{\perp,x},a^{\ast}_{\perp,y}\big]=\dlhat(x,y)~,\ 
\big[a^{\ast}_{\perp,x},a_{\perp,y}\big]=-\dlhat^{T}(x,y)~,\ 
\big[a_{\phibar},a^{\ast}_{\perp,x}\big]
=\big[a_{\perp,x},a^{\ast}_{\phi}\big]=0~.
\end{equation*}

\item[$\bullet$]

Fock space operators such as the full many-body Hamiltonian, $\mathcal H$, are primarily denoted by calligraphic letters. Some exceptions pertain to annihilation and creation operators including the field operators $a_x,\,a^\ast_x,\,\apx,\,\apx^\ast$ as well as the operators $a_{\overline \phi}, a^*_\phi$  and $a_j, a_j^{\ast}$ associated with the basis $\{e_j(x)\}_j\subset \mathfrak{h}$. 


\item[$\bullet$]

Functionals on Banach spaces are often denoted also by calligraphic letters. (Their distinction from Fock space operators is self evident.)   



\end{itemize}

\subsection{Paper organization}
\label{subsec:organiz}

The remainder of the paper is organized as follows. In Sect.~\ref{sec:results} we summarize our  results and approach. Section \ref{sec:Hamiltonian} focuses on the formal construction of the quadratic Hamiltonian, $\Hap$, and the derivation of the operator Riccati equation for the 
pair-excitation kernel. In Sect.~\ref{sec:Riccati} we develop an existence theory for this Riccati equation. In Sect.~\ref{sec:spectrum} we describe the excitation spectrum and construct the associated eigenvectors of $\Hap$. Key in this description is our use of the $N$-particle sector of the Bosonic Fock space. Section~\ref{sec:applications} addresses the intimate connection of our non-Hermitian theory for low-lying excitations to Fetter's approach~\cite{fetter72} and the properties of $J$-self-adjoint operators~\cite{Albeverio2009}. In Sect.~\ref{sec:conclusion} we conclude our paper by outlining a few open problems.

\section{Hamiltonian model, main results, and methodology}
\label{sec:results}
In this section, we introduce the full many-body Hamiltonian, and summarize our results and approach.  The more precise statements of results along with derivations or proofs can be found in the corresponding sections, as specified below.

The starting point is the many-body Hamiltonian in Fock space, viz., 
\begin{equation}\label{eq:Ham-def}
\Hcal=\int \dv x\dv y\,\left\{a^{\ast}_{x}\epsilon(x,y)a_{y}+\frac{1}{2}
 a^{\ast}_{x}a^{\ast}_{y} \upsilon(x-y)a_{x}a_{y}\right\}~,
\end{equation}
where $\epsilon(x,y)=\left\{-\Delta_{x}+V(x)\right\}\delta(x,y)$ is the kinetic part, $\upsilon(x)$ is the pairwise interaction potential, and $V(x)$ is the trapping potential. We assume that $\upsilon(x)$ is positive, symmetric, integrable and bounded on $\RR^3$. The trapping potential $V(x)$ is positive and such that the one-particle Schr\"odinger operator $-\Delta + V$ has discrete spectrum. 

\subsection{Reduced Hamiltonian and operator Riccati equation for $k$}
\label{subsec:reduced-Ricc-summ}

Section~\ref{sec:Hamiltonian} describes Wu's approach~\cite{wu98} in a language closer to operator theory, which serves our objectives.
By heuristics, we reduce Hamiltonian~\eqref{eq:Ham-def} to the quadratic form
\begin{equation*}
\Hap=N E_{\text{H}}+h(a_\perp^{\ast}, a_\perp)+\frac{1}{2N}f_\phi(a_\perp^{\ast},a_\perp^{\ast}) a_{\phibar}^2+\frac{1}{2N} \overline{f_\phi}(a_\perp,a_\perp) {a_\phi^{\ast}}^2~,
\end{equation*}
where $E_H$ is the (mean field) Hartree energy per particle, $h(a_\perp^{\ast}, a_\perp)$ and $\overline{f_\phi}(a_\perp,a_\perp)$ are operators of the form $\int \dv x\,\dv y\,\{a_{\perp, x}^\ast h(x,y) a_{\perp, y}\}$ and $\int \dv x\,\dv y\,\{a_{\perp, x} \overline{f_\phi(x,y)} a_{\perp, y}\}$ for suitable kernels $h(x,y)$ and $f_\phi(x,y)$, and $\phi$ is the condensate wave function; see Sect.~\ref{subsec:Herm}. Our derivation of the reduced Hamiltonian $\Hap$ relies on~\eqref{eq:ax-dec} and the conservation of the particle number. This $\Hap$ provides a simple model for pair excitation. Our goal is to solve the eigenvalue problem $\Hap \vert\psi\rangle = E_N \vert \psi\rangle$. 

Subsequently, we transform $\Hap$ non-unitarily according to $\tHap:= e^{\W}\Hap e^{-\W}$ where the operator $\W$ is of the form $-(2N)^{-1}\int \dv x\,\dv y\ \{k(x,y)\,a_{\perp,x}^\ast a_{\perp, y}^\ast\} a_{\phibar}^2$ which conserves the total number of particles; see Sect.~\ref{subsec:non-uni}.
The Riccati equation for kernel $k$ is extracted via the requirement that the non-Hermitian operator $\tHap$ does \emph{not} contain any terms with the product $a_\perp^\ast a_\perp^\ast$; see Sect.~\ref{subsec:Ricc-deriv}. 
If $k(x,\phibar)=0$, the Riccati equation for $k$ reads
\begin{equation*}
	h\circ k+ k\circ h^T +f_\phi  +k\circ \overline{f_\phi}\circ k=\lambda\otimes_{\text{s}}\phi~,
\end{equation*}
where the Lagrange multiplier $\lambda$ is determined self-consistently.

\subsection{Existence theory for $k$}
\label{subsec:existence-summ}
In Sect.~\ref{sec:Riccati}, we introduce a functional of $\kbar$ and $k$ by use of which we develop an existence theory for $k$. This functional, $\mathcal{E}[\kbar,k]:\mathrm{dom}(\mathcal{E})\to\mathbb{R}$, reads
\begin{equation*}
\mathcal{E}[\kbar, k]:=\mathrm{tr}\Big\{\big(\delta-\overline{k}\circ k\big)^{-1}\circ\Big(\overline{k}\circ h\circ k+\frac{1}{2}\overline{k}\circ f_\phi+\frac{1}{2}\overline{f_\phi}\circ k\Big)\Big\}~;
\end{equation*}
see Sect.~\ref{subsec:exist-lemmas} for the definition of $\mathrm{dom}(\mathcal{E})$. Setting the functional derivative of $\mathcal{E}[\kbar, k]$ with respect to $\kbar$ equal to zero yields the Riccati equation for $k$. 

We prove the existence of solutions to the Riccati equation for $k$ by assuming that 
\begin{align*}
h(\ebar,e)-\big\vert\vphi(\ebar,\ebar)\big\vert \geq 
c\Vert e\Vert_{L^{2}}^{2}\quad \quad
 \forall e\in \phi^{\perp}=
 \big\{e\in \mathfrak{h}^{1}_{V}\ \big\vert\ e\perp\phi\big\}~.
\end{align*}
In particular, this condition is satisfied if $\phi$ is a minimizer of the Hartree energy, $E_H$. 
The aforementioned inequality is employed as a hypothesis in the main existence theorem, Theorem~\ref{thm:Existence} (Sect.~\ref{subsec:exist-thm-proof}). In fact,  Theorem~\ref{thm:Existence} states that the above inequality and the property that $f_\phi$ is Hilbert-Schmidt imply that the functional $\mathcal{E}$ restricted to $\mathrm{dom}(\mathcal E)_\perp=\mathrm{dom}(\mathcal E)\cap \{k\in\mathfrak{B}_2(\mathfrak{h}^1_V)\, \big\vert\, k(x,\phibar)=0\}$ attains a minimum for some $k\in \mathrm{dom}(\mathcal E)_\perp$ which is a weak solution to the operator Riccati equation. We emphasize that $\phi$ does not need to be a minimizer of the Hartree energy. Our proof makes use of a basis of $\phi^\perp$, the theory of complex ($\mathcal{C}$-) symmetric operators and a variational principle based on functional $\mathcal E$. In Sect.~\ref{subsec:non-unique_k}, we discuss an implication of our existence theory, namely, the non-uniqueness of solutions to the Riccati equation.

\subsection{Spectrum and eigenvectors of reduced non-Hermitian Hamiltonian}
\label{subsec:reduced-spec-summ}
In Sect.~\ref{sec:spectrum}, we study the eigenvectors and spectrum of the non-unitarily transformed Hamiltonian $\tHap$, under the assumptions of Theorem~\ref{thm:Existence} for $k$. A highlight of our analysis is the explicit construction of these eigenvectors in $\FF_N$ by Fock space techniques. We write $\tHap=N E_H+\mathcal{\Hcal}_{\rm ph}$ where 
\begin{align*}
\mathcal{\Hcal}_{\rm ph}&:=
h_{\rm ph}\big(a^{\ast}_{\perp},a_{\perp}\big)
+\frac{1}{N}(a^{\ast}_\phi)^{2}
\vphib\big(a_{\perp},a_{\perp}\big)~; \\
  h_{\rm ph}\big(a^{\ast}_{\perp},a_{\perp}\big)&:=\int \dv x\,\dv y\ \{a_{\perp, x}^\ast (h+k\circ\vphib)(x,y) a_{\perp, y}\}~. 
\end{align*}
Evidently, $h_{\rm ph}(a^{\ast}_{\perp},a_{\perp})$ forms the diagonal part of $\tHap-N E_H$. We show that  $h_{\rm ph}$  is responsible for the discrete phonon-like excitation spectrum of the trapped Bose gas.

The main result is captured by a theorem (Theorem~\ref{thm:spectrum-phonon}), according to which the following equality of spectra holds:
\begin{equation*}
	\sigma\left(\Hcal_{\rm ph}\big\vert_{\FF_{N}}\right)
 =\sigma\left(h_{\rm ph}(a^{\ast}_{\perp},a_{\perp})\big\vert_{\FF_{N}}\right)~.
\end{equation*}
Furthermore, in this theorem we show that for every eigenvector of $h_{\text{ph}}(a_\perp^\ast, a_\perp)$ with eigenvalue $E$ there is a unique eigenvector of $\mathcal{\Hcal}_{\rm ph}$ with the same eigenvalue, $E$.

Our analysis is based on the following steps. First, we provide a formalism for the decomposition of $\FF_N$ into appropriate orthogonal subspaces (Sect.~\ref{subsec:spectrum-lemmas}). Our technique is similar to that in the construction by Lewin, Nam, Serfaty and Solovej~\cite{Lewin2014}. However, here we consider the eigenvectors of a Hamiltonian that conserves the number, $N$, of particles as opposed to the Bogoliubov Hamiltonian studied in~\cite{Lewin2014}. Second, we show that by the restriction $\Vert k\Vert_{\text{op}}<1$, the spectrum of the one-particle Schr\"odinger-type operator $h_{\text{ph}}$ is positive and discrete, and the corresponding eigenfunctions form a non-orthogonal Riesz basis of $\phi^\perp$ (Sect.~\ref{subsec:spectrum-hph}). The proof of the main theorem (Theorem~\ref{thm:spectrum-phonon}) relies on the above steps to show that the eigenvalue problem for $\mathcal{\Hcal}_{\rm ph}$ can be reduced to a finite-dimensional system of equations that has an upper triangular form (Sect.~\ref{subsec:proof-spectrum-thm}).    

\subsection{Connection of non-Hermitian and Hermitian approaches}
\label{subsec:connection-summ}
In Sect.~\ref{sec:applications}, we compare our approach to Fetter's formalism~\cite{fetter72}. In particular, we prove the existence of solutions to a PDE system for one-particle excitations, which reduces to Fetter's system~\cite{fetter72} when the interaction potential $\upsilon$ is replaced by $g\delta$ for some constant $g>0$. To this end, we assume that a solution to the operator Riccati equation exists. In this vein, we discuss the connection of the Riccati equation for $k$ to the theory of $J$-self-adjoint matrix operators by Albeverio and coworkers~\cite{AlbeverioMotovilov2019,Albeverio2009,AlbeverioMotovilov2010}. 

Starting with the relevant Bogoliubov Hamiltonian~\cite{fetter72}, we indicate that its diagonalization via ``quasiparticle'' operators (in Fetter's terminology) leads to the PDE system ($j=1,\,2,\,\ldots$)
\begin{equation*}
\begin{pmatrix}
 h^T_\perp & -{f_\phi}_\perp \\
 \overline{{f_\phi}}_\perp & -h_\perp
\end{pmatrix} 
\circ\begin{pmatrix}
u_j(x) \\
v_j(x) 
\end{pmatrix} = E_j
\begin{pmatrix}
u_j(x) \\
v_j(x)
\end{pmatrix}
\end{equation*}
for the one-particle wave functions $u_j$ and $v_j$ and respective eigenvalues $E_j$ (Sect.~\ref{subsec:Bog-eigenvalue}). Here, $q_\perp$ ($q=h,\,f_\phi$) is the projection of operator $q$ on space $\phi^\perp$. Notably, we show that the existence of solutions to the Riccati equation for $k$ implies the solvability of the above system for $(u_j, v_j)$; see Sect.~\ref{subsec:eigenv-Fet-existence}. We also prove that the completeness relations between $u_j$ and $v_j$, previously posed by Fetter~\cite{fetter72}, directly follow from our approach. In Sect.~\ref{subsec:J-self-adj}, we invoke ideas from $J$-self-adjoint operator theory to show that the restriction $\Vert k\Vert_{\text{op}}< 1$ yields a positive spectrum $\{E_j\}_{j=1}^\infty$ for the symplectic matrix involved in the system for $(u_j, v_j)$. 
\color{black}

\section{Construction of quadratic many-body Hamiltonian} 
\label{sec:Hamiltonian}
In this section, we formally construct a quadratic (Hermitian) Hamiltonian and transform it non-unitarily. A core ingredient of this approach is that the number of atoms is strictly conserved. We follow the treatment of Wu~\cite{wu61,wu98} but replace his delta-function potential for repulsive pairwise atomic interactions by a smooth potential. 

Section~\ref{subsec:Herm} focuses on heuristic approximations in the Hermitian setting, where we expand the Hamiltonian in powers of Boson field operators for noncondensate particles. Section~\ref{subsec:non-uni} concerns the non-unitary transformation of the quadratic Hermitian Hamiltonian. In Sect.~\ref{subsec:Ricc-deriv}, we derive a Riccati equation for the pair excitation kernel of the transformation. Section~\ref{subsec:Ricc-comm} provides some discussion on the procedure.

\subsection{Reduction of Hamiltonian in Hermitian setting}
\label{subsec:Herm}
In this subsection, we formally reduce the many-body Hamiltonian to a quadratic Hermitian operator in Fock space. The total number of particles is conserved. Our main result is described by \eqref{eq:H-app}--\eqref{eq:f-def} below.

We start with Hamiltonian~\eqref{eq:Ham-def}. 
Let $\phi$ denote the (one-particle) condensate wave function, which has $L^2$-norm $\Vert\phi\Vert_2=1$. Recall decomposition~\eqref{eq:ax-dec} for the Boson field operators $a_x$, $a_x^{\ast}$. The particle number operator, $\Ncal$, on $\FF$ can thus be decomposed as  
\begin{equation*}
	\Ncal=\int \dv x\ \{a_x^{\ast} a_x\}=a^{\ast}_\phi a_{\overline{\phi}}+\int \dv x\ \{\apx^{\ast} \apx\}=:\Ncal_\phi+\Ncal_\perp~,
\end{equation*}
where $\Ncal_\phi:=a^{\ast}_{\phi}a_{\overline{\phi}}$ is the number operator for condensate atoms; $\Ncal_\phi$ and $\Ncal_\perp$ commute, and $\Hcal$ commutes with $\Ncal$, viz., $[\Hcal, \Ncal]=\Hcal \Ncal-\Ncal \Hcal=0$. We use the $N$-th fiber, $\FF_N$, of the Bosonic Fock space, considering state vectors $\vert \psi\rangle_N$ that satisfy
\begin{equation*}
	\Ncal \vert \psi\rangle_N = N \vert \psi\rangle_N~;\qquad \Vert \vert \psi\rangle_N \Vert=1~.
\end{equation*}

Following Wu~\cite{wu61}, we first expand $\Hcal$ is powers of $\apx$, $\apx^{\ast}$ by applying decomposition~\eqref{eq:ax-dec} for $a_x$, $a_x^{\ast}$. The Hamiltonian $\Hcal$ reads
\begin{eqnarray*}
\Hcal&=&\int \dv x\,\dv y\ \left\{
\overline{\phi(x)}\epsilon(x,y)\phi(y)+
\frac{1}{2}(\Ncal_{\phi}-1)\vert\phi(x)\vert^{2}\upsilon(x-y)\vert\phi(y)\vert^{2}\right\}\Ncal_{\phi}
\\
&& \mbox{} +\int \dv x\,\dv y\ \left\{
\apx^{\ast}\Big(\epsilon(x,y)\phi(y)+
(\Ncal_{\phi}-1)\phi(x)\upsilon(x-y)\vert\phi(y)\vert^{2}
\Big) a_{\phibar}   \right\}
\\
&&\mbox{} +\int \dv x\,\dv y\ \left\{a^{\ast}_{\phi}
\Big(\overline{\phi(x)}\epsilon(x,y)+(\Ncal_{\phi}-1)
\overline{\phi(y)}\upsilon(x-y)\vert\phi(x)\vert^{2}\Big)
a_{\perp,y}\right\}
\\
&&\mbox{} +\int \dv x\, \dv y\ \left\{
\apx^{\ast}\Big(\epsilon(x,y)
+\Ncal_{\phi}\,(\upsilon\ast\vert\phi\vert^{2})(x)\,\delta(x,y)
+\Ncal_{\phi}\phi(x)\upsilon(x-y)\overline{\phi(y)}\Big)a_{\perp,y}
\right\}
\\
&&\mbox{} +\frac{1}{2}\int \dv x\,\dv y\ \left\{
\apx^{\ast} a^{\ast}_{\perp,y}\phi(x)\upsilon(x-y)\phi(y) a_{\phibar}^{2}
+
{a^{\ast}_{\phi}}^{2}
\overline{\phi(x)}\upsilon(x-y)\overline{\phi(y)}\apx a_{\perp,y}
\right\}
\\
&&\mbox{} +\int \dv x\,\dv y\ \left\{
\apx^{\ast} a^{\ast}_{\perp,y}
\upsilon(x-y)\phi(y)\apx a_{\phibar}
+a^{\ast}_{\phi}
\apx^{\ast} \overline{\phi(y)}\upsilon(x-y)\apx a_{\perp,y}\right\}
\\
&&\mbox{} +\frac{1}{2}\int \dv x\,\dv y\ \left\{
\apx^{\ast} a^{\ast}_{\perp,y}\upsilon(x-y)\apx a_{\perp,y}\right\}~.
\end{eqnarray*}
Recall that $\epsilon(x,y)=\left\{-\Delta_{x}+V(x)\right\}\delta(x,y)$.

The next step is to reduce $\Hcal$ to a Hermitian operator quadratic in $a_\perp$, $a_\perp^{\ast}$. First, we drop the terms that are cubic or quartic in $a_\perp$, $a^{\ast}_\perp$. Second, we make the substitution $\Ncal_\phi=\Ncal-\Ncal_\perp$ and replace $\Ncal$ by $N$ ($\Ncal\mapsto N$ with $N\gg 1$) because $\vert \psi\rangle\in\FF_N$. We then drop the term $\Ncal_\perp^2$. We take $N-1\simeq N$ and \emph{apply a Hartree-type equation} for the condensate wave function $\phi$ which we write as
\begin{equation*}
\int \dv y\ \left\{\epsilon(x,y)\phi(y)+N\phi(x)\upsilon(x-y)|\phi(y)|^2\right\}-\mu \phi(x)=0~.	
\end{equation*}
This results in the elimination of terms linear in $a_\perp$, $a_\perp^{\ast}$ in the Hamiltonian $\Hcal$. The multiplier $\mu$ enables us to impose the normalization constraint $\Vert\phi\Vert_2=1$; thus,
\begin{equation*}
\mu=\int \dv x\,\dv y\ \left\{\overline{\phi(x)}\epsilon(x,y)\phi(y)+N|\phi(x)|^2 \upsilon(x-y)|\phi(y)|^2\right\}~.	
\end{equation*}
The PDE for $\phi$ formally becomes the Gross-Pitaevskii equation~\cite{Gross61,Pitaevskii61} if $\upsilon$ is replaced by $g\delta$ for some constant $g>0$.  

Consequently, the original Hamiltonian $\mathcal H$ is reduced to the quadratic form
\begin{subequations}
\begin{equation}\label{eq:H-app}
\Hap=N E_{\text{H}}+h(a_\perp^{\ast}, a_\perp)+\frac{1}{2N}f_\phi(a_\perp^{\ast},a_\perp^{\ast}) a_{\phibar}^2+\frac{1}{2N} \overline{f_\phi}(a_\perp,a_\perp) {a_\phi^{\ast}}^2  
\end{equation}
where, abusing notation slightly, we define the operators
\begin{eqnarray}
	h(a_\perp^{\ast}, a_\perp)&:=&\int \dv x\,\dv y\ \left\{\apx^{\ast} h(x,y) a_{\perp,y}\right\}~,\label{eq:h-op-def}\\
	f_\phi(a_\perp^{\ast},a_\perp^{\ast})&:=& \int \dv x\,\dv y\ \left\{\apx^{\ast} f_\phi(x,y) a_{\perp,y}^{\ast}\right\}~, \label{eq:f-op-def}
	\end{eqnarray}
along with the corresponding kernels
\begin{eqnarray}
h(x,y)&:=&\epsilon(x,y)+N(\upsilon\ast|\phi|^2)(x)\,\delta(x,y)+N\gamma(x,y)-\mu~,\label{eq:h-def}\\
f_\phi(x,y)&:=& N\phi(x) \upsilon(x-y) \phi(y)~,\quad \color{black} \gamma(x,y):=\phi(x) \upsilon(x-y)\overline{\phi(y)}~.\label{eq:f-def} 	
\end{eqnarray}
\end{subequations}
In the above, the Hartree energy functional, $E_{\text{H}}$, is defined by
\begin{equation*}
E_{\text{H}}=\int \dv x\,\dv y\ \left\{\overline{\phi(x)}\epsilon(x,y)\phi(y)+\frac{N}{2}|\phi(x)|^2 \upsilon(x-y) |\phi(y)|^2\right\}~.
\end{equation*}

Equation~\eqref{eq:H-app} is the desired quadratic Hamiltonian. Note the key property
\begin{equation*}
[\Hap,\Ncal]=0~.	
\end{equation*}
%

\subsection{Non-unitary transformation of quadratic Hamiltonian $\Hap$}
\label{subsec:non-uni}
In this subsection, we transform $\Hap$ non-unitarily by use of the pair-excitation kernel, $k$. The main result is given by~\eqref{eq:H-transf-def} and~\eqref{eq:ric-def} below.

 For this purpose, we invoke the following quadratic operator:
\begin{subequations}
\begin{equation}\label{eq:K-op-def}
\K :=-\frac{1}{2}\int \dv x\,\dv y\ \left\{
k(x,y)a^{\ast}_{\perp,x}a^{\ast}_{\perp,y}\right\}~,
\end{equation}
where $k=k^T$. This $\K$ does not conserve the number of particles ($[\K, \Ncal]\neq 0$).  In addition, following Wu~\cite{wu61}, we introduce the operator
\begin{equation}\label{eq:W-op-def}
	\W:=-\frac{1}{2N}\int \dv x\,\dv y\ \{k(x,y)\apx^{\ast} a_{\perp,y}^{\ast}\} (a_{\phibar})^2=\frac{1}{N} \K (a_{\phibar})^2~.
\end{equation}
\end{subequations}
The kernel $k$ is not known at this stage, but must satisfy certain consistency conditions (see Sect.~\ref{subsec:Ricc-deriv}). We refrain from specifying the function space of $k$ now. A salient point of this formalism is the identity $[\W,\Ncal]=0$. Consequently, the operator $e^{\W}$, which is used to define the non-unitary transformation of $\Hap$ below, leaves $\FF_N$ invariant, i.e., $e^{\W}: \FF_N \mapsto \FF_N$. (However, $e^{\W}$ does not respect the Fock space norm.) Our goal here is to describe the non-Hermitian operator $e^{\W} \Hap e^{-\W}$.

The main idea concerning the proposed non-unitary transformation of $\Hap$ can be described as follows. Assume that $\vert \psi\rangle_N=\vert \psi\rangle$ ($\vert \psi\rangle\in \FF_N$) is an eigenvector of the (Hermitian) Hamiltonian $\Hap$ with eigenvalue $E$, viz.,
	$\Hap \vert\psi\rangle=E  \vert\psi\rangle$. Then, we have
\begin{equation*}
	\left\{e^{\W}\Hap e^{-\W}\right\} \big(e^{\W}\vert \psi\rangle\big)= E \big(e^{\W} \vert \psi\rangle\big)~.
\end{equation*}	
Hence, the non-Hermitian, non-unitarily transformed, operator $e^{\W}\Hap e^{-\W}$ has eigenvalue $E$ and eigenvector $e^{\W} \vert\psi\rangle$. It turns out that it is more tractable (in a certain sense, as shown below) to describe the transformed eigenvector $e^{\W} \vert\psi\rangle$ in $\FF_N$ than the original vector $\vert\psi\rangle$ by exploiting spectral properties of $e^{\W}\Hap e^{-\W}$. A price that one must pay for this option is that the pair-excitation kernel $k$ must satisfy the operator Riccati equation. One of our major goals here is to motivate the equation obeyed by $k$ through the computation of the non-Hermitian operator $e^{\W}\Hap e^{-\W}$. 

Next, we organize our calculation. First, we readily compute the conjugation
\begin{equation*}
	e^{\W} \apx e^{-\W}=\apx +\frac{1}{N} \khat^{T}(a^{\ast}_{\perp} ,x)\big(a_{\phibar}\big)^{2}~,
\end{equation*}
where (abusing notation) we define
\begin{eqnarray*}
\khat^{T}(x,y)&:=& \int \dv z\ \left\{k(x,z)\dlhat^{T}(z,y)\right\}~,\\
\khat^{T}(a^{\ast}_{\perp} ,x)
&:=&\int \dv y\,\dv z\ \left\{a^{\ast}_{\perp,y}
k(y,z)\dlhat^{T}(z,x)\right\}~.	
\end{eqnarray*}
In a similar vein, by virtue of~\eqref{eq:K-op-def} we compute
\begin{equation*}
	e^{\W} a_\phi^{\ast} e^{-\W}=a_\phi^{\ast} +\frac{2}{N}\K a_{\phibar}~.
\end{equation*}
In order to obtain a symmetric equation in the end, we symmetrize $h(a_\perp^{\ast}, a_\perp)$ as
\begin{equation*}
	h(a_\perp^{\ast}, a_\perp)=\frac{1}{2} \left\{h(a_\perp^{\ast}, a_\perp)+h^T(a_\perp,a_\perp^{\ast})\right\}+c_\infty~,
\end{equation*}
where $c_\infty$ is an (infinite) immaterial constant. This constant is harmless since it is added and subtracted. In fact, we remove this $c_\infty$ after we perform the calculation.

We proceed to carry out the computation of $e^{\W}\Hap e^{-\W}$. To avoid overly cumbersome expressions, we only display the manipulation of key terms of $\Hap$, for illustration purposes. We refrain from presenting the explicit computation of all terms. 

The main term that we need to compute reads
\begin{eqnarray*}
&&\int \dv x\, \dv y\ \left\{e^{\W}\Big((a^{\ast}_{\phi})^{2}a_{\perp,x}a_{\perp,y}\Big)e^{-\W}
\frac{1}{N}\overline{f_\phi(x,y)}\right\}\\
&=&
\left\{(a^{\ast}_{\phi})^{2}+\frac{2}{N}\K (2\Ncal_{\phi}-1)
+\frac{4}{N^2}\K^{2}(a_{\phibar})^{2}
 \right\}\\
 &&\times \int \dv x\,\dv y\ \left(a_{\perp,x}+\frac{1}{N}\khat^{T}(a^{\ast}_{\perp},x)
 (a_{\phibar})^{2}\right) \frac{1}{N}\overline{f_\phi(x,y)}
\left( a_{\perp,y}+\frac{1}{N}\khat(y,a^{\ast}_{\perp})
(a_{\phibar})^{2}\right)\\
&=& \left\{(a^{\ast}_{\phi})^{2}+\frac{2}{N}\K (2\Ncal_{\phi}-1)
+\frac{4}{N^2}\K^{2} (a_{\phibar})^{2} \right\} \\
&&\times \left\{\frac{1}{N}\overline{f_\phi}(a_{\perp},a_{\perp})
 +\frac{1}{N^2}\left((\khat^{T}\circ\overline{f_\phi})(a^{\ast}_{\perp},a_{\perp})
 +(\overline{f_\phi}\circ\khat)(a_{\perp},a^{\ast}_{\perp})\right)
(a_{\phibar})^2 \right. \\
&&\qquad \left. +\frac{1}{N^3}(\khat^{T}\circ \overline{f}\circ\khat)
 (a^{\ast}_{\perp},a^{\ast}_{\perp})(a_{\phibar})^{4}\right\}\\
  &=& (a^{\ast}_{\phi})^{2}\frac{1}{N}\overline{f_\phi}(a_{\perp},a_{\perp})
 +\frac{1}{N^2}\left\{(\khat^{T}\circ\overline{f_\phi})(a^{\ast}_{\perp},a_{\perp})
 +(\overline{f_\phi}\circ\khat)(a_{\perp},a^{\ast}_{\perp})\right\}
 \Ncal_{\phi}(\Ncal_{\phi}-1)\\
 &&\quad +\frac{1}{N^3}(\khat^{T}\circ \overline{f_\phi}\circ\khat)
 (a^{\ast}_{\perp},a^{\ast}_{\perp})
 (a_{\phibar})^{2}(\Ncal_{\phi}-2)(\Ncal_{\phi}-3)\\
 &&\qquad +{\rm higher\ order\ terms\ in\ } a_\perp,\,a_\perp^{\ast}~.
 \end{eqnarray*}
The above Fock space operator can be further simplified, without distortion of its commutability with $\Ncal$, via the replacement $\Ncal_\phi =\Ncal-\Ncal_\perp \mapsto N-\Ncal_\perp$.  	 
Subsequently, we drop terms higher than quadratic in $a_\perp$, $a_\perp^{\ast}$; and treat $N$ as large so that $N-l\simeq N$ if $l$ is fixed. The other relevant computations are 
\begin{eqnarray*}
	e^{\W} h(a_\perp^{\ast}, a_\perp) e^{-\W}&=& h(a_\perp^{\ast}, a_\perp)+\frac{1}{N} \big(h\circ \khat\big)(a_\perp^{\ast},a_\perp^{\ast})\,(a_{\phibar})^2~,\\
	e^{\W} h^T(a_\perp, a_\perp^{\ast}) e^{-\W}&=& h^T(a_\perp, a_\perp^{\ast})+\frac{1}{N} \big(\khat^T\circ h^T\big)(a_\perp^{\ast},a_\perp^{\ast})\,(a_{\phibar})^2~.
\end{eqnarray*}

Accordingly, we obtain the non-Hermitian quadratic operator
\begin{subequations}\label{eqs:H-tr-ric-def}
\begin{eqnarray}\label{eq:H-transf-def}
\tHap&:=& e^{\W} \Hap e^{-\W}= N E_{\text{H}}+\big(h +\khat^{T}\circ\overline{f_\phi}\big)
(a^{\ast}_{\perp} ,a_{\perp})+
\big(h^{T}+\overline{f_\phi}\circ \khat\big)(a_{\perp},a^{\ast}_{\perp}) \notag\\
&&\qquad +\frac{1}{N} \ric (a_\perp^{\ast},a_\perp^{\ast})\,(a_{\phibar})^2 +\frac{1}{N}(a_\phi^{\ast})^2\, \overline{f_\phi}(a_\perp,a_\perp)~.
\end{eqnarray}
In the formal limit $\upsilon\to \delta$, i.e., when the interaction potential becomes a delta function, this $\tHap$ becomes the reduced transformed Hamiltonian derived in~\cite{wu98}.  
The `Riccati kernel' is defined by
\begin{equation}\label{eq:ric-def}
\ric(x,y):=h\circ \dlhat \circ k+k\circ \dlhat^T \circ h^T +f_\phi +k\circ \dlhat^T \circ \overline{f_\phi}\circ \dlhat \circ k~.	
\end{equation}
\end{subequations}
Recall that the kernel $h(x,y)$ is defined by~\eqref{eq:h-def} with~\eqref{eq:f-def}, viz.,
\begin{equation*}
	h(x,y)=\{-\Delta_x +V(x)\}\delta(x,y)+N(\upsilon\ast|\phi|^2)(x)\delta(x,y)+N\phi(x)\upsilon(x-y)\overline{\phi(y)}-\mu~.
\end{equation*}

The operator $\tHap$ is the focus of our analysis. As we anticipated, we have the identity $[\tHap, \Ncal]=0$, which enables us to seek eigenvectors of $\tHap$ in $\FF_N$.  

\subsection{Riccati equation for $k$}
\label{subsec:Ricc-deriv}
Next, we heuristically outline the rationale for the derivation of an equation for $k$, in the spirit of Wu~\cite{wu61,wu98}. This equation is described by~\eqref{eq:k-Riccati} below. In Sects.~\ref{sec:Riccati}--\ref{sec:applications}, we rigorously study properties and implications of solutions to this equation.

By inspection of~\eqref{eq:H-transf-def}, we see that $\tHap-N E_H$ consists of  two types of terms: (i) Terms that contain both $a_\perp^{\ast}$ and $a_\perp$, and no $a_{\phibar}$ and $a_\phi^{\ast}$. The sum of these terms forms the `diagonal part' of $\tHap$, and can be described by use of a (nonlocal) one-particle Schr\"odinger operator. In the periodic setting~\cite{leehuangyang}, the use of this operator yields the phonon spectrum. (ii) Terms that contain $a_\perp^{\ast}$ and $a_{\phibar}$, or $a_\perp$ and $a_\phi^{\ast}$. In the periodic setting, it can be effortlessly argued that this second part does \emph{not} affect the phonon spectrum \emph{provided} $\ric(a_\perp^{\ast}, a_\perp^{\ast})=0$. Following Wu~\cite{wu61,wu98}, we require that
\begin{equation*}
	\ric=\lambda \otimes_{\text{s}}\phi~, 
\end{equation*}
where $\otimes_{\text{s}}$ denotes the symmetrized tensor product. In view of~\eqref{eq:ric-def}, we thus have an equation for $k$. Here, $\lambda(x)$ is arbitrary and can be chosen to satisfy a prescribed constraint involving the inner product $k(x,\phibar)$.  Notably, the operator $\W$ is invariant under changes of this constraint. In other words, physical predictions are not affected by the choice of $k(x,\phibar)$. For example, we can impose $k(x,\phibar)=0$~\cite{wu61,wu98}. This condition removes $\dlhat$, $\dlhat^T$ from the related equations, which is natural since
\begin{equation*}
\int \dv x\,\dv y\ \left\{\ric(x,y) \apx^{\ast}\apy^{\ast}\right\}=\int \dv x\,\dv y\ \left\{\big(\dlhat \circ \ric \circ \dlhat^T\big)(x,y)\ a_x^{\ast} a_y^{\ast}\right\}~. 	
\end{equation*}
The expression for $\ric(x,y)$ becomes
\begin{subequations}
\begin{equation}\label{eq:ric-nodelta}
\ric(x,y)=h\circ k+k\circ h^T +f_\phi +k\circ \overline{f_\phi}\circ k~.	
\end{equation}
Consequently, the equation for $k$ reads
\begin{equation}\label{eq:k-Riccati}
h\circ k+ k\circ h^T +f_\phi  +k\circ \overline{f_\phi}\circ k=\lambda\otimes_{\text{s}}\phi=\frac{1}{\sqrt{2}}(\lambda\otimes \phi+\phi\otimes\lambda)~, 	
\end{equation}
\end{subequations}
where $\lambda$ should be determined self-consistently. In fact, $\lambda(x)$ obeys the equation
\begin{equation*}
	\lambda(x)=C_1 \phi(x)+\sqrt{2}(h\circ k+ k\circ h^T +f_\phi +k\circ \overline{f_\phi}\circ k)(x,\phibar) =C_1 \phi(x)+\sqrt{2}(k\circ h^T +f_\phi)(x,\phibar)
\end{equation*}
with $C_1=-\langle \phibar,\lambda\rangle$; see Sect.~\ref{subsec:exist-thm-proof}. We refer to~\eqref{eq:k-Riccati} as the `operator Riccati equation' for $k$. In Sect.~\ref{sec:spectrum}, we show that this equation leads to an  excitation spectrum that is identical to the one from Fetter's formalism~\cite{fetter72}. By virtue of~\eqref{eq:k-Riccati}, the transformed approximate Hamiltonian becomes
\begin{eqnarray*}
\tHap&=& N E_{\text{H}}+\big(h +\khat^{T}\circ\overline{f_\phi}\big)
(a^{\ast}_{\perp} ,a_{\perp})+
\big(h^{T}+\overline{f_\phi}\circ \khat\big)(a_{\perp},a^{\ast}_{\perp})+\frac{1}{N}(a_\phi^{\ast})^2\, \overline{f_\phi}(a_\perp,a_\perp)~.
\end{eqnarray*}
%

\subsection{A few comments}
\label{subsec:Ricc-comm}

It is worthwhile to comment on aspects of our heuristic procedure. First, in hindsight, it is of some interest to discuss how~\eqref{eq:k-Riccati} can be motivated more transparently. The main observation is that, in regard to $\Hap$, we can consider the quadratic matrix form
\begin{equation*}
\int \dv x\,\dv y\ \big(a_{\perp,x}\ ,\ a^{\ast}_{\perp,x}\big)
\left(\begin{matrix}
-h^{T}(x,y)& N^{-1}(a^{\ast}_{\phi})^{2}\overline{f_\phi}(x,y)
\\
-N^{-1}f_\phi(x,y)(a_{\phibar})^{2}&h(x,y)
\end{matrix}\right)\left(\begin{matrix}
-a^{\ast}_{\perp,y}\\ a_{\perp,y}
\end{matrix}\right)~.
\end{equation*}
In view of the commutability of $a_\perp,\,a_\perp^{\ast}$ with $a_{\phibar},\,a^{\ast}_\phi$ we can perform the following conjugation of the above $2\times 2$ matrix, assuming for simplicity that $k(x,\phibar)=0$:
\begin{eqnarray*}
&&\left(\begin{matrix}
\delta & 0\\
N^{-1}k (a_{\phibar})^2  & \delta
\end{matrix}\right)\circ
\left(\begin{matrix}
-h^{T} & N^{-1} (a^{\ast}_{\phi})^{2}\overline{f_\phi}
\\
-N^{-1} f_\phi (a_{\phibar})^{2} & h
\end{matrix}\right)\circ
\left(\begin{matrix}
\delta & 0\\
-N^{-1}k (a_{\phibar})^{2}&  \delta
\end{matrix}\right)
\\
&=&\left(
\begin{matrix}
-h^{T}-N^{-2} \overline{f_\phi}\circ k \Ncal_{\phi}(\Ncal_{\phi}-1)
&N^{-1} (a^{\ast}_{\phi})^{2}\overline{f_\phi}
\\
-N^{-1} \widetilde{\rm Ric} (a_{\phibar})^{2}
&h+N^{-2} k\circ \overline{f_\phi}\Ncal_{\phi}(\Ncal_{\phi}-1)
\end{matrix}\right)
\end{eqnarray*}
where 
\begin{equation*}
\widetilde{\rm Ric}=
h\circ k+k\circ h^{T}+f_\phi 
+\frac{1}{N^2}k\circ\overline{f_\phi}\circ k\,\Ncal_{\phi}(\Ncal_{\phi}-1)~.
\end{equation*}
Now replace $\Ncal_\phi$ with $N$ in the last expression and take $N-1\simeq N$; thus, $\widetilde{\rm Ric}$ is reduced to $\ric$ (with the $\dlhat$ and $\dlhat^T$ removed). \emph{Equation~\eqref{eq:k-Riccati} then results from the requirement that the transformed $2\times 2$ matrix  is upper triangular.} We will show that this property implies that the excitation spectrum of $\Hap$ coincides with the one of the diagonal part of $\tHap$, and is identical to the spectrum of Fetter's approach~\cite{fetter72}; see Sect.~\ref{sec:spectrum}. 

A second comment concerns the Hartree-type equation for $\phi$, which becomes the Gross-Pitaevskii equation if $\upsilon$ is replaced by $g\delta$ for some constant $g>0$. We write the relevant PDE as $\HH_{\text{H}}\phi=\mu\phi$ where
\begin{equation}\label{eq:Hartree-1op-def}
\HH_{\text{H}}:=-\Delta_x +V(x)+N\big(\upsilon\ast |\phi|^2\big)(x)
\end{equation}
is a one-particle Hartree operator. In our analysis, we will consider the interaction potential, $\upsilon(x)$, to be positive, integrable and smooth. For a \emph{trapping} potential $V(x)$, where $V(x)\to \infty$ as $|x|\to \infty$, \color{black} the condensate wave function $\phi(x)$ is bounded and decays exponentially as $|x|\to\infty$. 

We are tempted to loosely comment on the assumptions underlying the uncontrolled approximations for the many-body Hamiltonian in this section.  We expect that the simplifications leading to the reduced Hamiltonian $\tHap$ make sense provided 
\begin{equation*}
\frac{\langle \psi\vert \Ncal_\perp^{l}\vert\psi\rangle}{N^l}\ll 1	\quad \forall\, \vert\psi\rangle\in\FF_N~;\quad l=1,\,2,\,3,\,4~.
\end{equation*}
%

\section{Existence theory for operator Riccati equation: Variational approach} 
\label{sec:Riccati}
In this section, we address the existence of solutions to~\eqref{eq:k-Riccati}. Our analysis is partly inspired by works of Albeverio, Tretter and coworkers, e.g., \cite{AlbeverioMotovilov2019,Albeverio2009,AlbeverioMotovilov2010,Tretter2016}, who rigorously connected the operator Riccati equation to the spectral theory of $J$-self-adjoint operators. In our work, we view an existence proof for $k$ as a necessary step towards ensuring the self-consistency of the approximation and non-unitary transform for the Bosonic many-body Hamiltonian. The existence proof for $k$ paves the way to establishing the connection of pair excitation to the phonon spectrum in a trap (Sect.~\ref{sec:spectrum}).

Our theory invokes an appropriate functional, $\mathcal E[\kbar, k]$, and two related lemmas (Sect.~\ref{subsec:exist-lemmas}).
A highlight is Theorem~\ref{thm:Existence} on the existence of $k$ (Sect.~\ref{subsec:exist-thm-proof}). We stress that our existence proof differs significantly from the approach found in~\cite{AlbeverioMotovilov2019,Albeverio2009,AlbeverioMotovilov2010}. First, we utilize a variational approach by seeking stationary points of the functional $\mathcal E[\kbar, k]$ on a Hilbert space, instead of applying the fixed-point argument of~\cite{Albeverio2009}. \color{black} Note that the fixed-point argument in~\cite{Albeverio2009} makes use of operator estimates that are not expected to hold for the operator $\mathrm{Ric}$ of~\eqref{eq:ric-def}. The variational approach developed here is amenable to constraints inherent to our problem; thus, the term $\lambda\otimes_{\text s}\phi$ of~\eqref{eq:k-Riccati} emerges as a Lagrange multiplier. Alternative approaches of variational character for block operator matrices (not for the Riccati equation per se) are described in~\cite{Tretter-book}. \color{black}

Second, our variational approach reveals that Riccati equation~\eqref{eq:k-Riccati} may in principle \emph{not} have a unique solution. Our existence proof indicates how one can construct an infinite number of solutions for $k$. These correspond to saddle points of the underlying functional, $\mathcal E$. This lack of uniqueness can pose a challenge in the subsequent analysis of the phonon spectrum (Sect.~\ref{sec:spectrum}). As a remedy to this issue, we point out that a restriction on the norm of $k$, i.e., $\|k\|_\mathrm{op}<1$, warrants uniqueness (see also~\cite{AlbeverioMotovilov2010}). By this restriction, the $k$ that solves Riccati equation~\eqref{eq:k-Riccati} is in fact a minimizer of $\mathcal E$.

\subsection{Functional $\mathcal E[\kbar, k]$ and useful lemmas}
\label{subsec:exist-lemmas}
Next, we define the relevant Hilbert space and the functional $\mathcal E[\kbar, k]$ which yields~\eqref{eq:k-Riccati}. We also prove two lemmas  needed for our existence theory.

\begin{definition}\label{def:energy-func} 
Let $\mathfrak{h}^1_V(\mathbb{R}^3\times\mathbb{R}^3)$ be the space of functions $k(x,x')$ such that 
\begin{equation*}
    \iint{\dv x\, \dv x'\ \Big\{|\nabla_x k(x,x')|^2+|\nabla_{x'}k(x,x')|^2+\big(V(x)+V(x')\big)|k(x,x')|^2\Big\}}<\infty~.
\end{equation*}
The energy functional  $\mathcal{E}[\kbar,k]:\mathrm{dom}(\mathcal{E})\to\mathbb{R}$ is defined by
\begin{subequations}\label{eqs:en-funct-dom}
\begin{equation} \label{eq:energy-func}
\mathcal{E}[\kbar, k]:=\mathrm{tr}\Big\{\big(\delta-\overline{k}\circ k\big)^{-1}\circ\Big(\overline{k}\circ h\circ k+\frac{1}{2}\overline{k}\circ f_\phi+\frac{1}{2}\overline{f_\phi}\circ k\Big)\Big\}
\end{equation}
where
\begin{equation}\label{eq:en-func-domain}
\mathrm{dom}(\mathcal{E}) := \Big\{k\in \mathfrak{B}_2(\mathfrak{h}^1_V)\big|\,k^T = k\,\,\mathrm{and}\,\, \Vert k\Vert_{\rm op}<1\Big\}\subset\mathfrak{B}_2(\mathfrak{h}^1_V)~.
\end{equation}
\end{subequations}
\end{definition}

\begin{remark}\label{rmk:hV1}
The space $\mathfrak{h}^1_V(\mathbb{R}^3\times\mathbb{R}^3)$ is the same as the space $\mathfrak{B}_2\big(\mathfrak{h}^1_V\big)$.	If $k\in\mathfrak{B}_2(\mathfrak{h}^1_V)$ and $\|k\|_2 <1$ then 
$(\delta-\overline{k}\circ k)^{-1}\in\mathfrak{B}_2(\mathfrak{h}^1_V)$. Thus, $\mathrm{dom}(\mathcal{E})$ is nonempty. The inequality $\|k\|_2 <1$ implies $\|k\|_\mathrm{op}<1$. Further remarks on $\|k\|_\mathrm{op}<1$ are deferred to Sect.~\ref{subsec:non-unique_k}. 
\end{remark}

The first lemma can be stated as follows.

\begin{lemma}\label{lem:variationalderivative}
The functional derivative of $\mathcal{E}[\overline{k},k]$ with respect to symmetric variations of $\overline{k}$ in $\mathfrak{h}^1_V(\mathbb{R}^3\times\mathbb{R}^3)$, denoted by $\delta\mathcal{E}/\delta\overline{k}$ where $\delta\mathcal{E}/\delta\overline{k}\in\mathfrak{B}_2^\ast(\mathfrak{h}^1_V)=\mathfrak{B}_2(\mathfrak{h}^1_V)$, is 
\begin{equation*}
    \frac{\delta\mathcal{E}[\kbar,k]}{\delta\overline{k}}= \frac{1}{2}(\delta-k\circ\overline{k})^{-1}\circ\big\{h\circ k+k\circ h^T+f_\phi+k\circ\overline{f_\phi}\circ k\big\}\circ(\delta-\overline{k}\circ k)^{-1}~.
\end{equation*}
\end{lemma}

\begin{proof}
Consider the arbitrary symmetric perturbation $\ell(x,\xprime)$.
It suffices to show that
\begin{align*}
&\Big(\frac{\dv}{\dv s}
\mathcal{E}[\kbar+s\overline{\ell},k]\Big)\Big\vert_{s=0}=
\\
&\frac{1}{2}\int \dv x\,\dv \xprime\ \left\{
\overline{\ell}(x,\xprime)(\delta-k\circ\overline{k})^{-1}\circ\big\{h\circ k+k\circ h^T+f_\phi+k\circ\overline{f_\phi}\circ k\big\}\circ(\delta-\overline{k}\circ k)^{-1}(x,\xprime)
\right\}\ .
\end{align*}
First, by differentiating the formal identity
$\delta = (\delta-\kbar \circ k)^{-1}\circ(\delta-\kbar\circ k)$ we obtain 
\begin{align*}
\Big(\frac{\dv}{\dv s}
\big\{\delta-(\kbar+s\overline{\ell})\circ k\big\}^{-1}\Big)\Big\vert_{s=0}
=\big(\delta-\kbar\circ k\big)^{-1}
\circ\overline{\ell}\circ k\circ\big(\delta-\kbar\circ k\big)^{-1}\ .
\end{align*}
Using, e.g., the Neumann series for $(\delta-\overline{k}\circ k)^{-1}$, we realize that
\begin{align*}
k\circ\big(\delta-\kbar\circ k\big)^{-1}
=\big(\delta-k\circ\kbar\big)^{-1}\circ k~.
\end{align*}
Hence, we also obtain the identity
\begin{align*}
\Big(\frac{\dv}{\dv s}\big\{\big(\delta-(\kbar+s\overline{\ell})\circ k\big)^{-1}
\circ(\kbar+s\overline{\ell})\big\}\Big)\Big\vert_{s=0}=
\big(\delta-\kbar\circ k\big)^{-1}\circ\overline{\ell}\circ
\big(\delta-k\circ\kbar\big)^{-1}~.
\end{align*}
Now express $\mathcal{E}$ as the sum
\begin{equation*}
\mathcal{E} = \mathrm{tr}\big\{(\delta-\overline{k}\circ k)^{-1}\circ\overline{k}\circ\big(h\circ k + {\textstyle\frac{1}{2}} f_\phi\big)\big\} + \mathrm{tr}\big\{\big(\delta-\overline{k}\circ k\big)^{-1}\circ{\textstyle\frac{1}{2}}(\overline{f_\phi}\circ k)\big\}=: \mathcal{E}_1 + \mathcal{E}_2~. 
\end{equation*}
The use of the cyclic property of the trace along with $\ell^T=\ell$ and $(\overline k\circ k)^T = k\circ\overline{k}$ yield
\begin{equation*}
\begin{split}
\Big(\frac{\dv }{\dv s}\mathcal{E}_1[\overline k +s\overline \ell,k]\Big)\Big|_{s=0} &= \mathrm{tr}\big\{(\delta-\overline k\circ k)^{-1}\circ\overline \ell\circ(\delta-k\circ \overline k)^{-1}\circ(h\circ k+{\textstyle\frac{1}{2}}f_\phi)\big\} \\
&= \mathrm{tr}\big\{\overline \ell\circ(\delta-k\circ \overline k)^{-1}\circ(h\circ k+{\textstyle\frac{1}{2}}f_\phi)\circ(\delta-\overline k\circ k)^{-1}\big\}
\end{split}
\end{equation*}
and
\begin{equation*}
\begin{split}
\Big(\frac{\dv}{\dv s}\mathcal{E}_2[\overline k +s\overline \ell,k]\Big)\Big|_{s=0} &= \mathrm{tr}\big\{{\textstyle\frac{1}{2}}(\delta-\kbar\circ k\big)^{-1}
\circ\overline{\ell}\circ k\circ\big(\delta-\kbar\circ k\big)^{-1}\circ(\overline{f_\phi}\circ k)\big\} \\
& = \mathrm{tr}\big\{{\textstyle\frac{1}{2}}
\overline{\ell}\circ k\circ\big(\delta-\kbar\circ k\big)^{-1}\circ(\overline{f_\phi}\circ k)\circ(\delta-\kbar\circ k\big)^{-1}\big\} \\
&=\mathrm{tr}\big\{{\textstyle\frac{1}{2}}
\overline{\ell}\circ\big(\delta-k\circ \kbar\big)^{-1}\circ(k\circ\overline{f_\phi}\circ k)\circ(\delta-\kbar\circ k\big)^{-1}\big\}~.
\end{split}
\end{equation*}
Now combine the above results to obtain the expression
\begin{equation*} 
\Big(\frac{\dv}{\dv s}\mathcal{E}[\overline k +s\overline \ell,k]\Big)\Big|_{s=0}=
\frac{1}{2}\mathrm{tr}\Big\{\overline{\ell}\circ\Big(
 \big(\delta-k\circ\kbar\big)^{-1}\circ \ric
 \circ\big(\delta-\kbar\circ k\big)^{-1}\Big)\Big\}~,
\end{equation*}
 where $\rm{Ric}$ is defined by~\eqref{eq:ric-nodelta}.
 Note that ${\rm Ric}$ is manifestly symmetric
 if $k$ is symmetric.
 This observation completes the proof of Lemma~\ref{lem:variationalderivative}. \hfill $\square$
 \end{proof}
 
 \begin{remark}\label{rmk:weak-soln}
The notion of the weak solution as the critical point of the functional $\mathcal E[\kbar,k]$ is relevant to our existence theorem (Theorem~\ref{thm:Existence}). Consider the space $\phi^{\perp}=
 \big\{e\in \mathfrak{h}^{1}_{V}\ \big\vert\ e\perp\phi\big\}$. We remind the reader that a bounded operator $k\in\mathfrak{B}(\phi^{\perp},\phi^{\perp})$ has a weak solution to the Riccati equation
\begin{equation*}
   k\circ h_\perp^T + h_\perp\circ k + k\circ\overline{f_\phi}\circ k +f_\phi = 0 
    \end{equation*}
provided
\begin{equation*}
    \langle k\circ h_\perp^T p,r\rangle +\langle k p,h^T_\perp r\rangle+\langle k\circ\overline{f_\phi}\circ k p,r\rangle =\langle-f_\phi p,r\rangle.\quad\forall p,\,r\in\mathrm{dom}\{h^T_\perp\}~,
\end{equation*} 
where $h_\perp$ is the projection of operator $h$ on space $\phi^\perp$. 
\end{remark}

Before stating the second lemma, we remark on the condensate wave function, $\phi$.

\begin{remark}\label{rmk:Hartree}
By Sect.~\ref{sec:Hamiltonian}, recall that $\phi$ satisfies $\HH_{\rm H}\phi(x) = \mu\phi(x)$  where the one-particle Hartree operator $\HH_{\rm H}$ is defined in~\eqref{eq:Hartree-1op-def}. We now state a few assumptions, which primarily concern the interaction potential $\upsilon(x)$ and the trapping potential $V(x)$.
First, let us assume that $\upsilon(x)$ is positive, symmetric, integrable, and smooth. \color{black}  If the equation for $\phi$ comes from minimizing the Hartree energy functional, $E_{H}$, viz., 
$$
E_{H}(\phi)
:=\int \dv x\, \dv y\ \left\{
\overline{\phi(x)}\epsilon(x,y)\phi(y)+\frac{N}{2}\vert\phi(x)\vert^2\upsilon(x-y)\vert\phi(y)\vert^2\right\}~,
$$
with $\Vert\phi\Vert_2=1$ then $\mu$ is
the lowest eigenvalue of the linear operator that results from fixing $\phi$ in $\HH_{\rm H}$. The existence theorem (Theorem~\ref{thm:Existence}) is stated and proved for a condensate $\phi$ that is not necessarily a minimizer of $E_H$. In fact, we replace the assumption of $\phi$ being such a minimizer by a less restrictive hypothesis (see Lemma~\ref{lem:interaction}). \color{black} We assume that the potential $V$ is such that $-\Delta+V$ has 
 discrete spectrum; for example, $V(x) = c\vert x\vert^2$ ($c>0$). The spectrum of  
 $\HH_{\rm H}$ is also discrete since $\HH_{\rm H}$ is a compact perturbation of $-\Delta+V$. \color{black} 
 \end{remark}
 %

\begin{lemma}\label{lem:interaction}
If $\upsilon(x)$ has positive Fourier transform $\widehat{\upsilon}(\xi)$, $\widehat{\upsilon}(\xi)\ge 0$, and $\phi$ is a minimizer of the functional $E_H(\phi)$, then for some $c>0$
the following inequality holds:
\begin{equation*}
     h(\ebar,e)-\big\vert\vphi(\ebar,\ebar)\big\vert\geq c\Vert e\Vert_{2}
     \quad \forall e\in \phi^\perp=\big\{e\in \mathfrak{h}_V^1\color{black}\ \big\vert\ e\perp\phi\big\} 
\end{equation*}
where $h(\cdot,\cdot)$ and $\vphi(\cdot,\cdot)$ are defined from~\eqref{eq:h-op-def}--\eqref{eq:f-def}. 
\end{lemma}
%

\begin{proof}
Define $g(x):= \overline{\phi(x)}\,e(x)$. Parseval's identity yields
\begin{equation*}
    \iint{\dv x\,\dv y\ \{\overline{e(x)}\phi(x)N\upsilon(x-y) \overline{\phi}(y)e(y)\}} = \int{\dv\xi\ \{N\widehat{\upsilon}(\xi)|\widehat{g}(\xi)|^2\}}~,
\end{equation*}
which dominates the integral
\begin{equation*}
    \iint{\dv x\,\dv y\ \{\overline{e}(x)f_\phi (x,y)\overline{e}(y)\}} =\int{d\xi\{N\widehat{\upsilon}(\xi)\big(\overline{\widehat{g}(\xi)}\big)^2}\}~.
\end{equation*}
Since $\phi$ is the minimizer of the Hartree functional, $E_H(\phi)$, we can assert that
$\phi$ is the  eigenfunction  with the lowest eigenvalue of
the operator $\HH_{\rm H}$ and is therefore simple.
If $e\perp\phi$ then $\big<\ebar, \HH_{H}e\big>\geq c\Vert e\Vert_{L^{2}}$
for some $c>0$ because the spectrum of the Hartree operator $\HH_{H}$ is discrete (see Remark~\ref{rmk:Hartree}). \hfill $\square$
\end{proof}

Lemma~\ref{lem:interaction} motivates the inequality involving $h$ and $f_\phi$ as a key assumption of Theorem~\ref{thm:Existence}, which replaces the requirement that $\phi$ is a minimizer of $E_H(\cdot)$.

\subsection{Existence theorem and proof}
\label{subsec:exist-thm-proof}
The existence theorem can be stated as follows:
\begin{theorem}\label{thm:Existence}
Suppose that the kernels $h(x,y)$ and $\vphi(x,y)$ satisfy the inequality
\begin{align}
h(\ebar,e)-\big\vert\vphi(\ebar,\ebar)\big\vert \geq 
c\Vert e\Vert_{L^{2}}^{2}\quad \quad
 \forall e\in \phi^{\perp}=
 \big\{e\in \mathfrak{h}^{1}_{V}\ \big\vert\ e\perp\phi\big\}~,
 \label{eq:Condition1}
\end{align}
for some constant $c>0$. Moreover, let us assume that $\vphi$ is Hilbert-Schmidt.

Consider the functional $\mathcal{E}[\kbar,k]$, defined in~\eqref{eq:energy-func}, 
with domain 
\begin{equation*}
\mathrm{dom}(\mathcal{E})_\perp := \mathrm{dom}(\mathcal{E})
\cap\left\{k\in\mathfrak{B}_2(\mathfrak{h}^1_V)\ \big\vert\ k(x,\phibar)=0\right\}
\end{equation*}
which consists of the compact $\mathcal{C}$-symmetric Hilbert-Schmidt operators $k\in\mathfrak{B}_2(\mathfrak{h}^1_V)$ satisfying $k(x,\overline{\phi})=0$. 

Then the functional $\mathcal{E}$ restricted to $\mathrm{dom}(\mathcal{E})_\perp$ attains a minimum for some $k\in\mathrm{dom}(\mathcal{E})_\perp$ which is a weak solution of the operator Riccati equation~\eqref{eq:k-Riccati}. 
The function $\lambda(x)$ entering this equation is a Lagrange multiplier due to the restriction and equals
\begin{equation}\label{eq:lagrangemultiplier}
  \lambda(x) = \sqrt{2}\left\{\big(k\circ\gamma\big)\big(x, \overline{\phi}\big)+f_\phi\big(x,\overline{\phi}\big)-{\textstyle\frac{1}{2}}\vphi(\phibar,\phibar)\phi(x)\right\}~. 
\end{equation}
\end{theorem}

At this stage, two remarks are in order.

\begin{remark}\label{rmk:C-symm}
We will seek stationary points of $\mathcal{E}[\kbar,k]$ under the constraint $k^T = k$. We now describe a generalization of the spectral theorem for compact operators with \emph{symmetric} kernels which is invoked in the proof of Theorem~\ref{thm:Existence}. Let $\mathcal{C}$ denote the operator of complex conjugation on $\mathfrak{h}$ where
\begin{equation*}
\mathcal{C}f(x)=\overline{f(x)} \qquad  \forall\,f\in \mathfrak{h}~.
\end{equation*}
An operator $\mathcal T$ on $\mathfrak{h}$ \color{black} is called \textit{complex-symmetric} (``$\mathcal{C}$-symmetric'') if it satisfies 
\begin{equation*}
\mathcal{C}\mathcal{T}=\mathcal{T}^\ast\mathcal{C}~,
\end{equation*}
where $\mathcal{T}^\ast$ is the Hermitian conjugate of $\mathcal{T}$ ($\mathcal{T}^\ast(x,y)=\overline{\mathcal{T}(y,x)}$). Clearly, integral operators whose kernels are symmetric in their arguments are $\mathcal{C}$-symmetric. 
An important property is that any compact complex-symmetric operator $\mathcal T$ such that $\mathcal{T}^\ast\circ \mathcal{T}$ has \emph{simple spectrum} admits the decomposition
\begin{equation}\label{eq:csymmetric}
\mathcal{T} = \sum_{n=1}^\infty{a_n (u_n\otimes \mathcal{C}u_n)}~,
\end{equation}
where $a_n\in\mathbb{C}$ converge to zero as $n\to\infty$ and $\{u_n\}_{n=1}^\infty$ is an orthonormal basis of $\mathfrak{h}$. This property comes from the identity $(\mathcal{C}\mathcal{T})\circ(\mathcal{C}\mathcal{T}) = \mathcal{T}^\ast\circ \mathcal{T}$, which implies the commutation relation
$[\mathcal{C}\mathcal{T}, \mathcal{T}^\ast \circ \mathcal{T}] = 0$.
In particular, $\mathcal{C}\mathcal{T}$ commutes with the spectral measure (and any eigenprojector) of the positive operator $\mathcal{T}^\ast \circ \mathcal{T}$. Since the latter operator has simple spectrum, it follows that these two operators have the same eigenspace. This fact allows us to pass from the eigenvalue equation $(\mathcal{T}^\ast \circ \mathcal{T})(u_n,x) = c_n u_n(x)$ to the eigenvalue equation $\mathcal{C}\mathcal{T}(u_n,x) = a_n u_n(x)$; 
thus, $|a_n|^2 = c_n$.
It can be directly shown that all $\mathcal{C}$-symmetric tensor products $u\otimes v$ \color{black} must have $v(x) = \mathcal{C}u(x)$. Hence,  we can also pass from the spectral representation 
\begin{equation*}
\mathcal{T}^\ast\circ \mathcal{T} = \sum_{n=1}^\infty{c_n (u_n\otimes u_n)}
\end{equation*}
to expression~\eqref{eq:csymmetric}. This result amounts to a version of the spectral theorem for compact $\mathcal{C}$-symmetric operators; see, e.g.,~\cite{Garcia2005}. 
\end{remark}
\begin{remark}
The reader should compare~\eqref{eq:lagrangemultiplier}, regarding the Lagrange multiplier $\lambda$, with equation (3.24) in Wu's paper~\cite{wu98}, which employs a delta-function interaction potential. The respective formulas for $\lambda(x)$ differ by a factor of $\sqrt{2}$ because of our choice of a different normalization factor for $\lambda\otimes_{\mathrm s} \phi$. 	
\end{remark}

We can now proceed to prove Theorem~\ref{thm:Existence}. Notably, we consider a condensate $\phi$ that is not necessarily a minimizer of the Hartree energy, $E_H$.

%
%

\begin{proof} We split the proof of Theorem~\ref{thm:Existence} into three main steps.

{\bf Step 1.} We now express the functional $\mathcal E$ in terms of a suitable basis and describe critical points, by taking into account the theory of $\mathcal C$-symmetric operators. By Remark~\ref{rmk:C-symm}, any $k$ satisfying our assumptions admits the decomposition
\begin{equation*}
k(x,x') =\sum_{j=1}^\infty {z_{j}e_{j}(x)e_{j}(x')}\ ,
\quad  e_{j}\in\phi^{\perp}~,
\end{equation*}
where $\{e_{j}(x)\}_j\subset \mathfrak{h}$ is an orthonormal basis and the coefficients $\{z_{j}\}\subset\mathbb{C}$ are such that $z_{j}\to0$ as $j\to \infty$. For the moment, we assume $|z_j|\not=1$ for all $j$ so that
\begin{equation*}
    (\delta-\overline{k}\circ k)^{-1}(x,\xprime)= \color{black} \sum_{j= 1}^\infty{\Big(\frac{1}{1-|z_j|^2}\Big)\ebar_j(x)e_j(x')}~.
\end{equation*}
The substitution of the two preceding expressions into~\eqref{eq:energy-func} for the energy furnishes
\begin{equation*}
\mathcal{E}\Big(\{e_j\}, \{z_{i}\}\Big)	=\sum_{j= 1}^\infty {\frac{1}{1-|z_{j}|^{2}}\left\{
h(\ebar_{j},e_{j})|z_{j}|^{2}+\frac{1}{2}\Big(
\vphi(\ebar_{j},\ebar_{j})z_j+\vphib(e_{j},e_{j}) \overline{z}_j\Big)\right\}}~,
\end{equation*}
where $\vphib(e_{j},e_{j})=\overline{\vphi(\ebar_{j},\ebar_{j})}$. \color{black} The derivative of $\mathcal{E}\big(\{e_j\},\{z_j\}\big)$ with respect to $\overline{z_{j}}$ reads
\begin{equation*} 
\frac{\partial}{\partial\overline{z}_{j}}\mathcal{E}\Big(\{e_{j}\},\{z_j\}\Big)=
\frac{1}{2}\sum_{j=1}^{\infty}
\frac{2h(\ebar_{j},e_{j})z_{j}+\vphib(e_{j},e_{j})+
\vphi(\ebar_{j},\ebar_{j})z_{j}^{2}}{(1-\vert z_{j}\vert^{2})^{2}}~. \color{black}
\end{equation*}
Setting $\partial\mathcal{E}/\partial\overline{z}_{j}=0$ gives two roots, viz., 
\begin{equation} \label{eq:alphai}
z^{\pm}_j=\frac{-h(\ebar_{j},e_{j})\pm\sqrt{h^{2}(\ebar_{j},e_{j})-
\vert\vphi(\ebar_{j},\ebar_{j})\vert^{2}}}{\vphi(\ebar_{j},\ebar_{j})}~.
\end{equation}
The assumption stated by~\eqref{eq:Condition1} 
guarantees that $\vert z^{\pm}_{j}\vert\not= 1$, provided $e_j$ is a member of the function space $\phi^\perp$; in fact, $|z_j^+|< 1$ and $|z_j^-|>1$. Regarding $\mathcal E(\{e_j\}, \{z_j\})$, notice that the summand (for fixed $j$ and $e_j=e$) is described by the function
\begin{equation*}
f(z;e):= \color{black}\frac{2h(\ebar,e)\vert z\vert^{2}+\vphi(\ebar,\ebar)\overline{z}
+\vphib(e,e)z}{1-\vert z\vert^{2}}
\end{equation*}
which takes real values with $f(0;e)=0$, while
\begin{align*}
\lim_{\vert z\vert\to 1^{-}}f(z;e)=+\infty~.
\end{align*}
Thus, the function $f(z;e_j)$
attains a  minimum at $z=z_{j}^{+}$. \color{black} On the other hand, we have
\begin{align*}
\lim_{\vert z\vert\to\infty}f(z;e)
=-2h(\ebar,e)~,
\end{align*}
and $f(z_j^{-};e_j)=-h(\ebar_j,e_j)-
\sqrt{h^{2}(\ebar_j,e_j)-\vert\vphi(\ebar_j,\ebar_j)\vert^{2}}$
which implies that $f(z;e_j)$  has a maximum at $z=z_j^{-}$ in view of
\begin{equation*}
\lim_{\vert z\vert\to 1^{+}}f(z;e)=-\infty~.
\end{equation*}

By~\eqref{eq:alphai} the evaluation of $\mathcal{E}$ with the roots $z^{\pm}_j$  yields 
\begin{equation*} 
\mathcal{E}\Big(\{(z^{\pm}_{j}\}, \{e_j\}\Big)= 
-\frac{1}{2}\sum_{j=1}^{\infty}
{\Big\{h(\ebar_{j},e_{j})\mp 
\sqrt{h^{2}(\ebar_{j},e_{j})-|\vphi(\ebar_{j},\ebar_{j}|^2}\Big\}}~.
\end{equation*} 
In light of the preceding discussion, we choose the root $z^{+}_{j}$ where $\vert z^{+}_{j}\vert <1$ and define
\begin{align*}
\mathcal{F}(e):=h(\ebar,e)-
\sqrt{h^{2}(\ebar,e)-\vert\vphi(\ebar,\ebar)\vert^{2}}~,
\end{align*}
so that the value of the functional $\mathcal E$ reads
\begin{align*}
\mathcal{E}\big(\{z_{j}^{+}\},\{e_{j}\}\big)
=-\frac{1}{2}\sum_{j=1}^{\infty}\mathcal{F}(e_{j})\ .
\end{align*}

{\bf Step 2.}  
So far, the minimization problem over compact symmetric operators has been converted into the following problem: 
\begin{align*}
\min_{\{e_{j}\},{\rm orthonormal}}\left\{-\frac{1}{2}\sum_{j=1}^{\infty}
\mathcal{F}(e_{j})\right\}=-\frac{1}{2}\max_{\{e_{j}\},{\rm orthonormal}}
\sum_{j=1}^{\infty}\mathcal{F}(e_{j})~.
\end{align*}
Next, we prove that the minimum is attained. The construction of the orthonormal set  $\{e_{j}(x)\}_{j=1}^\infty$
can be carried out inductively via a standard procedure, 
so we skip the details here. 
For simplicity, we describe the first
step. The remaining steps are similar. Our task is to maximize $\mathcal{F}(e)$, focusing on
\begin{align*}
\max_{\Vert e\Vert_{L^{2}} =1,e\in\phi^\perp}\mathcal{F}(e)~.
\end{align*}
Consider a maximizing sequence $\{e_{n}\}\in \phi^\perp$ with $\Vert e_{n}\Vert_{L^{2}}=1$
such that
\begin{align*}
\lim_{n\to\infty}\mathcal{F}(e_{n})=:\mathcal{F}_{\rm max}~. \color{black}
\end{align*}
We can assert that $\mathcal{F}_{\rm max}$ is finite, which follows from the
observation
\begin{align*}
\mathcal{F}(e)=
\frac{\vert\vphi(\ebar,\ebar)\vert^{2}}
{h(\ebar,e)+\sqrt{h^{2}(\ebar,e)-\vert\vphi(\ebar,\ebar)\vert^{2}}}
\leq \vert\vphi(\ebar,\ebar)\vert\leq \Vert\vphi\Vert_{L^{2}}
\Vert e\Vert_{L^{2}}^{2}~.
\end{align*} 
We can assume without loss of generality that $\mathcal{F}_{\rm max}>0$,
for if $\mathcal{F}_{\rm max}=0$ then 
$\vphi(\ebar,\ebar)$ vanishes identically on the set on which 
we try to maximize, namely
$\phi^\perp$.
We know that $\vphi(\ebar_{n},\ebar_{n})$ is bounded and from the
expression of $\mathcal{F}(e_{n})$ we conclude that 
$h(\ebar_{n},e_{n})$ is also bounded. Since 
$\{e_{n}\}\subset \phi^{\perp}$ we can find a subsequence 
(again denoted by $\{e_{n}\}$) which converges weakly in $\phi^{\perp}$
to some $e_{1}\in \phi^{\perp}$. Because $\phi^{\perp}$ is compactly
embedded in $\mathfrak{h}=L^2(\mathbb{R}^3)$ we conclude that 
$\{e_{n}\}$ converges strongly in $\mathfrak{h}$.
Up to this subsequence, we therefore have
\begin{equation*}
\lim_{n\to\infty}\mathcal{F}(e_{n})=\mathcal{F}(e_{1})=
\mathcal{F}_{\rm max}~.
\end{equation*}
Furthermore, we can assert that 
\begin{equation*}
h(\ebar_{n},e_{n})\to h(\ebar_{1},e_{1})~.
\end{equation*}
The reason is that $\mathcal{F}(e)$ is a decreasing
function of $h(\ebar,e)$. Thus, if 
$$
h(\ebar_{n},e_{n})\to \widetilde{h}<h(\ebar_{1},e_{1})
$$
then $\lim_{n\to\infty}\mathcal{F}(e_{n})$ will not be the maximum,
which leads to a contradiction.
In conclusion, the subsequence $\{e_{n}\}$ converges 
strongly in $\phi^{\perp}$ and 
$\mathcal{F}(e_{1})=\mathcal{F}_{\rm max}$.

Next, we check that the overall minimum is finite. Condition~\eqref{eq:Condition1} implies that $h(\ebar,e)$ is bounded below. This 
means that the the overall minimum is finite provided 
\begin{align*}
\sum_{j=1}^{\infty}\mathcal{F}(e_{j})
&\leq C\sum_{j=1}^{\infty}
\vert\vphi(\ebar_{j},e_{j})\vert^{2}
=\sum_{j=1}^{\infty}\left\vert\ 
\int
\dv x\,\dv y\,\left\{\ebar_{j}(x)\vphi(x,y)\ebar_{j}(y)\right\}\right\vert^{2}
\\
&\leq \int\dv x\,
\sum_{j=1}^{\infty}\left\vert\ \int
\dv y\,\left\{\vphi(x,y)\ebar_{j}(y)\right\}\right\vert^{2}
\leq \int 
\dv x\,\dv y\,\left\{\vert\vphi(x,y)\vert^{2}\right\}<\infty~.
\end{align*}
The last condition indeed holds. Note that $z_j^+$ is recast to the expression
\begin{align*}
z_{j}^{+}=\frac{\vphi(\ebar_{j},\ebar_{j})}
{h(\ebar_{j},e_{j})+\sqrt{h^{2}(\ebar_{j},e_{j})-\vert\vphi(\ebar_{j},\ebar_{j})
\vert^{2}}}~.
\end{align*}
Accordingly, we see that
\begin{align*}
\sum_{j=1}^{\infty}\vert z_{j}^{+}\vert^{2}\leq C
\sum_{j=1}^{\infty}\vert\vphi(\ebar_{j},\ebar_{j})\vert^{2}
\leq \Vert\vphi\Vert^{2}_{L^{2}(\mathbb{R}^{3}\times\mathbb{R}^{3})}~.
\end{align*}
Thus, $k$ is Hilbert-Schmidt.

{\bf Step 3.} So far, we showed that the minimum is attained in $\mathrm{dom}(\mathcal{E})_\perp$. We must now take into account the constraint
 $k(x,\phibar)=0$ via a Lagrange multiplier. We introduce the
 Lagrange multiplier as an operator with 
 symmetric kernel $\ell(x,y)$ where
 \begin{align*}
 \ell(x,y)=\big(\lambda\otimes_{\rm s}\phi\big)(x,y)
 :=\frac{1}{\sqrt{2}}\left\{\lambda(x)\phi(y)+\lambda(y)\phi(x)\right\}~. 
 \end{align*}
In the above, $\phi$ is the condensate wave function and $\lambda(x)$ is to be
determined. 
Hence, the modified energy functional to be minimized has the form 
\begin{align*}
\tilde{\mathcal{E}}[\kbar, k]:=\mathcal{E}[\kbar,k]-{\rm tr}
\left\{\overline{\ell}\circ k+\kbar\circ \ell\right\}~.
\end{align*}
In view of Lemma~\ref{lem:variationalderivative}, setting equal to zero the functional derivative of $\tilde{\mathcal{E}}$ with respect to $\kbar$ yields Riccati equation~\eqref{eq:k-Riccati}.
%
%
Given that $h=\mathbb{H}_{\rm H}\delta +N\gamma-\mu$ and
$\mathbb{H}_{\rm H}\phi=\mu\phi$, we compute $\lambda$ by
contracting the above equation for $k$ with $\phibar$. Thus, we obtain~\eqref{eq:lagrangemultiplier}.

{\em Note on $k$ as a weak solution.} We conclude the proof by showing that $k$ satisfies the definition of a weak solution (Remark~\ref{rmk:weak-soln}).  The condition that $\mathcal{E}$ is minimized implies that the first variation with respect to $\kbar$ vanishes, i.e.,
\begin{equation*}
    \mathrm{tr}\Big[\kbar_1\circ(\delta-k\circ\overline{k})^{-1}\circ\Big\{h_\perp\circ k+k\circ h^T_\perp+f_{\phi}+k\circ\overline{f_\phi}\circ k\Big\}\circ(\delta-\overline{k}\circ k)^{-1}\Big]=0~,
\end{equation*}
for all $\kbar_1\in\mathrm{dom}(\mathcal{E})$. Without loss of generality, we can make the substitution $\kbar_1 \mapsto (\delta-\overline{k}\circ k)^{-1}\circ\kbar_1\circ(\delta-\overline{k}\circ k)^{-1}$ so that the equation for vanishing first variation reads
\begin{equation*}
    \mathrm{tr}\Big[\kbar_1\circ\Big\{h_\perp\circ k+k\circ h^T_\perp+f_\phi+k\circ\overline{f_\phi}\circ k\Big\}\Big]=\mathrm{tr}\Big[\kbar_1\circ\mathrm{Ric}\Big]=0~.
\end{equation*}
If $\kbar_1 := p\otimes _{\mathrm s}r$ for $p,\,r\in\mathrm{dom}(h_\perp)$, the condition $\mathrm{tr}\Big[\kbar_1 \circ\mathrm{Ric}\Big]=0$ \color{black} translates to 
\begin{equation*}
\begin{split}
\mathrm{tr}\Big[p\otimes_{\mathrm{s}} r\circ\mathrm{Ric}\Big] &= \mathrm{tr}\Big[(k\circ p\otimes_{\mathrm{s}} r)\circ h_\perp \Big]+\mathrm{tr}\Big[(p\otimes_{\mathrm{s}} r)\circ  k\circ h_\perp\Big] \\ 
&+\mathrm{tr}\Big[(p\otimes_{\mathrm{s}} r)\circ\Big(f_\phi+k\circ\overline{f_\phi}\circ k\Big)\Big] \\
&= \sqrt{2} \big(\langle k p, h_\perp r\rangle + \langle p, (k\circ h_\perp) r\rangle + \langle p, \big(f_\phi+k\circ\overline{f_\phi}\circ k\big)r\rangle\big)~.
\end{split}
\end{equation*}
Observe that this expression is defined for any $p,\,r\in\mathrm{dom}(h_\perp)$, which is the space of test functions for the weak formulation of the Riccati equation.
\hfill $\square$
\end{proof}
%

\subsection{On the non-uniqueness of solution for $k$}
\label{subsec:non-unique_k}

Next, we discuss the important issue of the non-uniqueness of solutions for $k$ by our variational approach. The energy functional $\mathcal{E}[\kbar,k]$ has infinitely many critical points which correspond to choosing one of the two roots $z_j^\pm$, at every index $j$, in the proof of Theorem~\ref{thm:Existence}. We can show that these different choices correspond to minimax points. To this end, pick an
arbitrary $e_{1}(x)$, normalized so that $\Vert e_{1}\Vert_{\mathfrak{h}}=1$, and
let 
\begin{align*}
X^{\perp}(\phi,e_{1}):=\left\{e(x)\ \big\vert\ e\perp \{\phi,e_{1}\}\right\}~.
\end{align*}
Also, consider the subspace 
$X(e_{1}):={\rm span}(e_{1})=\left\{z_{1}e_{1}(x)\in\mathfrak{h}\ \vert\ z_{1}\in{\mathbb C}\right\}$.
Accordingly, we set up the min-max problem expressed by
\begin{align*}
\max_{e\in X(e_{1}),\vert z_{1}\vert >1}\left\{
\min_{\Vert k\vert_{X^{\perp}(\phi,e_{1})}\Vert_{\rm op}<1}
\mathcal{E}(\kbar,k)\right\}~.
\end{align*}
By repetition of the above argument (proof of Theorem~\ref{thm:Existence}), this min-max problem generates the (saddle type) critical point of the functional $\mathcal{E}(\kbar,k)$. More generally, the maximum can be taken over any finite collection of $\{(e_{j_k},z_{j_k})\}_k$, producing a unique solution for every distinct sequence.

Evidently, the only solution $k$ that obeys $\|k\|_\mathrm{op}<1$ is the one given in the proof of Theorem~\ref{thm:Existence}. Thus, we single out the choice $\{z_j = z_j^+\}$ with $|z_j^+|<1$ for all $j$ as the one yielding the unique pair-excitation kernel for our model. 
\begin{remark}\label{rmk:positivity}
 If $z_j$ is chosen in~\eqref{eq:alphai} such that $z_j = z_j^+$ for all $j> 0$ (i.e., $\|k\|_\mathrm{op}<1$) then
\begin{equation*}
\begin{split}
    (h+k\circ\overline{f_\phi})(\overline{e}_j,e_j) &= h(\overline{e}_j,e_j) + (k\overline{e}_j,\overline{f_\phi}e_j)= h(\overline{e}_j,e_j)+(z_j^{+})\overline{f_\phi}(e_j,e_j) \\
    &= \sqrt{h^2(\overline{e}_j,e_j)-|f_\phi(\overline{e}_j,\overline{e}_j)|^2}>0~.
\end{split}
\end{equation*}
This property will be relevant for the spectrum of the reduced Hamiltonian (Sect.~\ref{subsec:spectrum-hph}).
\end{remark}

\section{Spectrum and eigenvectors of reduced Hamiltonian} 
\label{sec:spectrum}
In this section, we describe the spectrum and eigenvectors of the reduced transformed Hamiltonian $\tHap$ in the $N$-th sector of Fock space, $\FF_N$  (see Sect.~\ref{subsec:Ricc-deriv}). For this purpose, we decompose $\FF_N$ into suitable orthogonal subspaces. A similar technique is used in~\cite{Lewin2014} in connection to the Bogoliubov Hamiltonian which does not conserve the particle number.

We start by writing the transformed, quadratic non-Hermitian Hamiltonian as
\begin{subequations}\label{eqs:H-phon-def}
\begin{align}
\tHap=NE_H+\mathcal{\Hcal}_{\rm ph}~,\quad \mathcal{\Hcal}_{\rm ph}:=
h_{\rm ph}\big(a^{\ast}_{\perp},a_{\perp}\big)
+\frac{1}{N}(a^{\ast}_\phi)^{2}\,
\vphib\big(a_{\perp},a_{\perp}\big)~\label{eq:Hphon-F} 
\end{align}
where $\vphi(a_{\perp}^\ast,a_{\perp}^\ast)$ (and thus $\vphib(a_{\perp},a_{\perp})$) is defined by~\eqref{eq:h-op-def} with~\eqref{eq:h-def}, and
\begin{align}
h_{\rm ph}:=h+k\circ\vphib~. \label{eq:hphon}
\end{align}
\end{subequations}
The operator $h_{\rm ph}\big(a^{\ast}_{\perp},a_{\perp}\big)$ forms the diagonal part of $\Hcal_{\rm ph}$ and is non-Hermitian. 
 
 The main result of this section is expressed by the following theorem.
 %
 \begin{theorem}\label{thm:spectrum-phonon}
 Consider the operators $\Hcal_{\rm ph}$ and $h_{\rm ph}(a^{\ast}_{\perp},a_{\perp})$
 restricted on $\FF_N$. Then 
 \begin{align*}
 \sigma\left(\Hcal_{\rm ph}\big\vert_{\FF_{N}}\right)
 =\sigma\left(h_{\rm ph}(a^{\ast}_{\perp},a_{\perp})\big\vert_{\FF_{N}}\right)~.
 \end{align*}
 Moreover, for each eigenvector $\vert\Omega\rangle_{N}\in \FF_N$ 
 of $h_{\rm ph}(a^{\ast}_{\perp},a_{\perp})$ with eigenvalue $E$ there exists
 a unique eigenvector $\vert\Psi(\Omega)\rangle_{N}\in \FF_N$ of
 $\Hcal_{\rm ph}$ such that
 \begin{align*}
 \Hcal_{\rm ph}\vert\Psi(\Omega)\rangle_{N}
 =E\vert\Psi(\Omega)\rangle_{N}~.
 \end{align*}
 \end{theorem}
 %

 Before we give a proof of Theorem~\ref{thm:spectrum-phonon}, we need to provide a few useful results. In Sect.~\ref{subsec:spectrum-lemmas}, we develop a formalism for the decomposition of $\FF_N$ into orthogonal subspaces. A key ingredient of our approach is the use of Fock space techniques. In Sect.~\ref{subsec:spectrum-hph}, we show that the spectrum of $h_{\rm ph}$ is discrete. In Sect.~\ref{subsec:proof-spectrum-thm}, we use this machinery (theory of Sects.~\ref{subsec:spectrum-lemmas} and~\ref{subsec:spectrum-hph}) to prove Theorem~\ref{thm:spectrum-phonon}. In our proof, we describe an explicit construction of the 
 eigenvectors of $\Hcal_{\rm ph}$ restricted on $\FF_{N}$ in terms of eigenvectors of $h_{\rm ph}(a_\perp^\ast, a_\perp)$. This construction invokes the discrete spectrum of $h_{\rm ph}$. 
 
\subsection{Decomposition of $\FF_N$ and two related lemmas} 
\label{subsec:spectrum-lemmas}
Next, we set the stage for the proof of Theorem~\ref{thm:spectrum-phonon}. We use the symbol $\vert\psi\rangle_n$ to denote the
vector  of $\FF$ with entry $\psi_{n}(x_{1},\,\ldots,\, x_{n})$ in the $n$-th slot and zero elsewhere. Let us also introduce the projection
$\Pcal_{N}\ :\ \FF\mapsto \FF_{N}$. \color{black}

The vector $\psi_{N}(x_{1},\,\ldots,\, x_{N})$ can be decomposed
as a direct sum according to
\begin{align*}
\psi_{N}=\sum_{n=0}^{N}\psi_{N,n}\otimes_{\rm s}
\big(\otimes^{N-n}\phi\big)=:\sum_{n=0}^{N}\psi^{\phi}_{N,n}~,\qquad \otimes^{p}\phi:=\prod_{j=1}^{p}\phi(x_{j})~,
\end{align*}
where the vectors $\psi_{N,n}(x_{1},\,\ldots\,, x_{n})$ 
satisfy the orthogonality relations ($n=1,2\ldots N$)
\begin{align*}
&\int \dv x\ \left\{\phibar(x)\psi_{N,n}(x,\,x_{2},\,\ldots\,, x_{n})\right\}=0~.
\end{align*}
We also define $\vert\psi^{\phi}_{N,n}\rangle_N:=
\big(0,\,\ldots,\, 0,\,\psi^{\phi}_{N,n},\,\ldots\big)\in \FF_{N}$.
Hence, the vector $\vert\psi^{\perp}\rangle
:=\big(\psi_{N,0},\,\psi_{N,1},\,\ldots,\, \psi_{N,N},\,0,\,\ldots\big)$
describes fluctuations around the tensor product 
$\otimes^{N}\phi$ (pure condensate). This means that we can decompose the space $\FF_{N}$ into the following direct sum 
of orthogonal subspaces: 
\begin{equation*}
\FF_{N}=\oplus_{n=0}^{N}\FF_{N,n}~;\qquad  \FF_{N,n}={\rm span}_{\psi_{n}\perp\phi}
\big(0,\,\ldots,\,0,\,\psi_{n}\otimes_{\mathrm{s}}(\otimes^{N-n}\phi),\,0\ldots\big)~. 
\end{equation*}

To describe this decomposition, we consider the number operator $\Ncal_{\phi}=a^*_\phi a_{\bar\phi}$ for the condensate. We have 
\begin{align*}
\Ncal_{\phi}\big\vert_{\FF_{N,n}}=(N-n)\I\big\vert_{\FF_{N,n}}
\quad ,\quad 0\leq  n\leq N~,
\end{align*}
where $\I$ is the identity operator on $\FF$. \color{black}
In other words, $\FF_{N,n}$ are 
the eigenspaces of $\Ncal_{\phi}$ restricted to $\FF_{N}$.
To see how the decomposition works, we invoke the identity 
\begin{align*}
\Pcal_{N}=\sum_{n=0}^{N}\frac{(-1)^{N-n}}{(N-n)!n!}
\prod_{p=0\atop p\not= N-n}^{N}\big(\Ncal_{\phi}-p\I\big)
\Pcal_{N}~,
\end{align*}
which can be understood as a resolution of the identity on $\FF_{N}$.
By introducing the projection operator (projection on $\FF_{n}$)
via the polynomial
$$
P_{n}(z):=\frac{(-1)^{n}}{n!}\prod_{j=1}^{n}
(z-j)~,
$$
we define
\begin{align*}
\Pcal_{n,n}:=P_{n}(\Ncal_{\phi})=\frac{(-1)^{n}}{n!}
\prod_{j=1}^{n}\big(\Ncal_{\phi}-j\I\big)~.  
\end{align*}
Thus, $\Pcal_{n,n}$ is the projection $\Pcal_{n,n}:\FF_{n}\mapsto\FF_{n,n}$.

At this stage, we can state the first lemma of this section as follows.

\begin{lemma}\label{lem:identity-P}
The operator $\Pcal_{n,n}=P_n(\Ncal_\phi)$ satisfies the factorization \color{black}
\begin{equation*}
\sum_{n=0}^{N}\frac{(-1)^{N}}{(N-n)!}\,{a^{\ast}_{\phi}}^{N-n}
\Pcal_{n,n} a_{\phibar}^{N-n}
=\sum_{n=0}^{N}\frac{(-1)^{N-n}}{(N-n)!\,n!}
\prod_{p=0\atop p\not= N-n}^{N}\big(\Ncal_{\phi}-p\I\big)
\end{equation*}
which, restricted to $\FF_{N}$,
describes a resolution of $\FF_{N}$ into orthogonal
subspaces.
\end{lemma}

\begin{proof}
First, we note the useful identities
\begin{align*}
{a_{\phibar}}^{p}\Ncal_{\phi}=\big(\Ncal_{\phi}+p\I\big)
{a_{\phibar}}^{p}~,\quad 
{a^{\ast}_{\phi}}^{p}\Ncal_{\phi}
=\big(\Ncal_{\phi}-k\I\big){a^{\ast}_{\phi}}^{p}\qquad (p=0,\,1,\,\ldots N)~.
\end{align*}
Subsequently, for any polynomial $P(z)$ we can assert that
\begin{align*}
P(\Ncal_{\phi}) {a_{\phibar}}^{p}={a_{\phibar}}^{p}P\big(\Ncal_{\phi}-p\I\big)~,\quad 
P(\Ncal_{\phi}){a^{\ast}_{\phi}}^{p}=
{a^{\ast}_{\phi}}^{p}P\big(\Ncal_{\phi}+p\I\big)~.
\end{align*}
Moreover, we have the following formulas: 
\begin{align*}
a_{\phibar}^{n}{a^{\ast}_{\phi}}^{n}
=\prod_{p=1}^{n}\big(\Ncal_{\phi}+p\I\big)~,\quad 
{a^{\ast}_{\phi}}^{n}a_{\phibar}^{n}
=\prod_{p=0}^{n-1}\big(\Ncal_{\phi}-p\I\big)~.
\end{align*}
By using the above relations, we write
\begin{align*}
&\sum_{n=0}^{N}\frac{(-1)^{N}}{(N-n)!}{a^{\ast}_{\phi}}^{N-n}
P_{n}(\Ncal_{\phi})a_{\phibar}^{N-n}=\sum_{n=0}^{N}
\frac{(-1)^{N}}{(N-n)!}
{a^{\ast}_{\phi}}^{N-n}
{a_{\phibar}}^{N-n}
P_{n}\big(\Ncal_{\phi}-(N-n)\I\big)
\\
&=\sum_{n=0}^{N}
\frac{(-1)^{N}(-1)^{n}}{(N-n)!n!}
\prod_{j=0}^{N-n-1}\big(\Ncal_{\phi}-j\I\big)
\prod_{p=1}^{n}
\big(\Ncal_{\phi}+(n-p-N)\I\big) \\
&=\sum_{n=0}^{N}
\frac{(-1)^{N-n}}{(N-n)!\,n!}
\prod_{j=0\atop j\not=N-n}^N\big(\Ncal_{\phi}-j\I\big)=\Pcal_{N}~.
\end{align*}
We can show that $\Pcal_{n,n}$ is a projection. Indeed, notice that 
$\Ncal_{\phi}P_{n}(\Ncal_{\phi})\big\vert_{\FF_{n}}=0$,
and if $\vert\psi\rangle_{n}\in\FF_{n}$ then 
$\Pcal_{n,n}\vert\psi\rangle_{n}=\vert\psi_{n,n}\rangle_{n}$.
If we consider the decomposition
\begin{align*}
\psi_{n}=\sum_{p=0}^{n}\psi_{n,p}\otimes_{s}\big(\otimes^{n-p}\phi\big)
=\sum_{p=0}^{n}\psi^{\phi}_{n,p}
\end{align*}
then $\Pcal_{n,n}$ applied to $\big\vert\psi\big>_{n}$
produces a vector in $\FF_{n,n}$, where $\FF_{n,p}$  is the decomposition of $\FF_{n}$ for
$p=0,1\ldots n$. \hfill $\square$
\end{proof}

For later algebraic convenience, we give the following definition (cf.~\cite{Lewin2014}).

\begin{definition}\label{def:UUstar}
Consider the operators $\U_{n},\,\U^{\ast}_{n}:\FF_{N}\mapsto \FF_{N}$ 
given by 
\begin{align*}
\U_{n}:=\Pcal_{n,n}\frac{a_{\phibar}^{N-n}}{\sqrt{(N-n)!}}~,\quad 
\U_{n}^{\ast}:=
\frac{{a^{\ast}_{\phi}}^{N-n}}{\sqrt{(N-n)!}}
\Pcal_{n,n}~;\quad   n=0,\,1,\,\ldots,\, N~.
\end{align*}
\end{definition}

By Definition~\ref{def:UUstar}, the result of Lemma~\ref{lem:identity-P} implies that 
\begin{align*}
\Pcal_{N}=\sum_{n=0}^{N}
\frac{(-1)^{N}}{(N-n)!}
{a^{\ast}_{\phi}}^{N-n}\Pcal_{n,n}\Pcal_{n.n}
{a_{\phibar}}^{N-n}\Pcal_{N}&=(-1)^{N}
\sum_{n=0}^{N}\frac{{a^{\ast}_{\phi}}^{N-n}\Pcal_{n,n}}{\sqrt{(N-n)!}}
\frac{\Pcal_{n,n}{a_{\phibar}}^{N-n}}{\sqrt{N-n)!}}\Pcal_{N}
\\
&=(-1)^{N}\sum_{n=0}^{N}\U_{n}^{\ast}\U_{n}\Big\vert_{\FF_{N}}~.
\end{align*}
The key relations following from this decomposition are 
\begin{align*}
\U_{n}\vert\psi\rangle_{N}=\vert\psi_{N,n}\rangle_{n}~,\quad 
\U^{\ast}_{n}\vert\psi_{N,n}\rangle_{n}=
\vert \psi_{N,n}^{\phi}\rangle_{N}\quad {\rm and}
\quad \sum_{n=0}^{N}\vert\psi^{\phi}_{N,n}\rangle_N=\vert\psi\rangle_{N}~.
\end{align*}
Hence, we have $\U_{n}:\FF_{N}\mapsto \FF_{n,n}$ and $\U^{\ast}_{n}:\FF_{n}\mapsto \FF_{N,n}$. \color{black}

\begin{lemma}\label{lem:U-prod} 
The operators $\{\U_{n}\}_{n=0}^{N}$ (see Definition~\ref{def:UUstar}) satisfy the relation
\begin{align*}
\U_{m}\U^{\ast}_{n}=\delta_{m,n}\Pcal_{m,m}\Pcal_{n,n}~.
\end{align*}
\end{lemma}

\begin{proof}
First, we make the observation that
\begin{align*}
\U_{m}\U^{\ast}_{n}=\Pcal_{m,m}
\frac{{a_{\phibar}}^{N-m}}{\sqrt{(N-m)!}}
\frac{{a^{\ast}_{\phi}}^{N-n}}{\sqrt{(N-n)!}}
\Pcal_{n,n}~.
\end{align*}
If $m<n$ (thus, $N-m>N-n$), in view of the property $a_{\phibar}\Pcal_{n,n}\big\vert_{\FF_{n}}=0$ we have 
\begin{align*}
\U_{m}\U^{\ast}_{n}
=\frac{\prod_{j=n-m}^{N-m-1}
\big(\Ncal_{\phi}+j\I\big)}
{\sqrt{(N-m)!(N-n)!}}\Pcal_{m,m}{a_{\phibar}}^{n-m}
\Pcal_{n,n}=0~.
\end{align*}
In this vein, if $n<m$ then $\U_{m}\U^{\ast}_{n}=0$. 
By $\Ncal_{\phi}\Pcal_{n,n}=0$ we assert that if $m=n$ then
\begin{align*}
\U_{n}\U^{\ast}_{n}=
\Pcal_{n,n}\frac{\prod_{j=1}^{N-n}\big(\Ncal_{\phi}+j\I\big)}{(N-n)!}
\Pcal_{n,n}=\Pcal_{n,n}~.\hskip2in \square
\end{align*}
\end{proof} 
%

\subsection{On the spectrum of $h_{\rm ph}$} 
\label{subsec:spectrum-hph}
Next, we discuss key properties of the diagonal part, $h_{\rm ph}(a_\perp^*, a_\perp)$, of the reduced Hamiltonian. Interestingly, $h_{\rm ph}(x,y)$ is similar to a self-adjoint operator. As such, many important spectral properties of self-adjoint operators carry over to $h_{\rm ph}$. \color{black}

\begin{lemma}\label{lem:hphon-spectrum} 
Assume that the pair-excitation kernel $k$ solves the operator Riccati equation~\eqref{eq:k-Riccati} with $\|k\|_\mathrm{op}<1$. 

\textbf{(i)} Then the spectrum of $h_\mathrm{ph}\,:\,\mathfrak{h}^1_V\cap \phi^\perp \mapsto\mathfrak{h}$ \color{black} is real and discrete.  The corresponding eigenfunctions $\omega_j(x)$, which satisfy $h_\mathrm{ph}(x,\omega_j)=E_j\omega_j(x)$ where $E_j>0$ are the eigenvalues (for $j=1,\,\dots$),     
form a non-orthogonal Riesz basis of $\phi^\perp$. 

\textbf{(ii)} Also, suppose that the functions $u_j(x)$ solve the adjoint problem,  i.e.,  
$h_{\rm ph}^\ast(x,u_j) = E_j u_j(x)$ on $\phi^\perp$ (for $j=1,2,\dots$).
Then the following completeness relation holds:
\begin{equation}\label{eq:completeness-omega}
\sum_{j=1}^\infty{\overline{\omega_j(x)}u_j(y)}=\widehat\delta(x,y)~.
\end{equation} 
\end{lemma}
%

\begin{proof}
The following relation holds on $\phi^\perp$ by use of the Riccati equation~\eqref{eq:k-Riccati}:
\begin{equation}\label{eq:k-proj-proof}
   (h+k\circ \overline{f_\phi})\circ (\widehat \delta-k\circ\overline{k}) = (\widehat \delta-k\circ\overline{k})\circ (h+f_\phi\circ\overline{k})~.
\end{equation}
If $\|k\|_\mathrm{op}<1$ then $(\widehat\delta-k\circ\overline k)^{-1}$ and $(\widehat\delta-k\circ\overline k)^{1/2}$ exist and are bounded operators on $\phi^\perp$. \color{black} 
By~\eqref{eq:k-proj-proof}, the operator $\varkappa:=(\widehat\delta-k\circ\overline{k})^{-1/2}\circ(h+k\circ\overline{f_\phi})\circ(\widehat\delta-k\circ\overline{k})^{1/2}$ is self-adjoint. Recall that $h$ has discrete spectrum; thus, $h+k\circ \overline{f_\phi}$ has discrete spectrum because $k\circ \overline{f_\phi}$ is compact. Moreover, the eigenvalues of $h+k\circ \overline{f_\phi}$ are positive (see Remark~\ref{rmk:positivity}). 
By the spectral theorem, the eigenvalues of $\varkappa$ are then positive and discrete, and the respective eigenvectors form an orthonormal basis of $\phi^\perp$. \color{black}

\textbf{(i)} Let the eigenvalues of $\varkappa$ be $\{E_j\}_{j=1}^\infty$, with eigenvectors $\{\eta_j\}_{j=1}^\infty$. 
Since the mapping $(h+k\circ\overline{f_\phi})\mapsto(\widehat\delta-k\circ\overline{k})^{-1/2}(h+k\circ\overline{f_\phi})(\widehat\delta-k\circ\overline{k})^{1/2}$ is a similarity transformation, the operator $h_{\rm ph}$ also has real spectrum $\{E_j\}_{j=1}^\infty$. From the relation 
\begin{subequations}\label{eqs:basis-omega-u}
\begin{equation}\label{eq:omega-basis}
\omega_j(x) = (\widehat\delta-k\circ\overline{k})^{1/2}(x,\eta_j)~, 
\end{equation}
 we conclude that $\{\omega_j(x)\}_{j=1}^\infty$ forms a Riesz basis as a bounded perturbation of an orthonormal basis of $\phi^\perp$. 

\textbf{(ii)} In a similar vein, the family $\{u_j(x)\}_{j=1}^\infty$ defined by
\begin{equation}\label{eq:u-basis}
    u_j(x):= (\widehat\delta-k\circ\overline k)^{-1/2}(x,\eta_j) 
\end{equation}
\end{subequations}
forms a Riesz basis for the adjoint problem on $\phi^\perp$. 

The resolution of the identity by the eigenvectors $\eta_j$ of the operator $\varkappa$ reads
\begin{equation*}
    \sum_{j=1}^{\infty}{\eta_j(x)\overline{\eta_j(y)}} = \widehat\delta(x,y)~.
\end{equation*}
This  equation yields completeness relation~\eqref{eq:completeness-omega}, by use of~\eqref{eqs:basis-omega-u}. \hfill $\square$
\end{proof}


\subsection{Proof of Theorem~\ref{thm:spectrum-phonon}}
\label{subsec:proof-spectrum-thm}
We are now in position to prove Theorem~\ref{thm:spectrum-phonon}. 
Our argument for the construction of eigenvectors of $\Hcal_{\rm ph}$ relies on the fact the
 $h_{\rm ph}$ has discrete spectrum. Let $\vert\psi\rangle:=\vert\psi\rangle_N$.

\begin{proof}
{\bf Step 1.} We decompose $\Hcal_{\rm ph}\big\vert\psi\big>$, where 
$\Hcal_{\rm ph}$ is given in~\eqref{eqs:H-phon-def}. We show that  
\begin{align}\label{eq:Hphon-dec}
&\sum_{m=0}^{N}\U^{\ast}_{m}\U_{m}\Hcal_{\rm ph}
\sum_{n=0}^{N}\U^{\ast}_{n}\U_{n}\vert\psi\rangle 
=\sum_{m,n=0}^{N}\U^{\ast}_{m}
\left\{\U_{m}\Hcal_{\rm ph}\U^{\ast}_{n}\right\}\left(\U_{n}\big\vert\psi\rangle\right)
\nonumber
\\
&=\sum_{n=0}^{N}
\U^{\ast}_{n}h_{\rm ph}(a^{\ast}_{\perp},a_{\perp})\U_{n}\vert\psi\rangle
+\sum_{n=0}^{N-2}b_{N,n}\U^{\ast}_{n}\vphib(a_{\perp},a_{\perp})\U_{n+2}
\vert\psi\rangle~.  
\end{align}
Here, $b_{N,n}$ is a numerical constant. Regarding the operators $\U_n$, see  Definition~\ref{def:UUstar}.

In order to derive~\eqref{eq:Hphon-dec}, we invoke Lemma~\ref{lem:U-prod}. In this vein, we notice the relations
\begin{align*}
h_{\rm ph}(a_{\perp}^{\ast},a_{\perp})\vert\psi\rangle=\sum_{n,m=0}^{N}\U^{\ast}_{n}
\U_{n}
h_{\rm ph}(a^{\ast}_{\perp},a_{\perp})\U^{\ast}_{m}\U_{m}\vert\psi\rangle
=\sum_{n=0}^{N}\U_{n}^{\ast}h_{\rm ph}(a_{\perp}^{\ast},a_{\perp})
\U_{n}\vert\psi\rangle~, \color{black}
\end{align*}
\begin{align*}
&\frac{{a^{\ast}_{\phi}}^{2}}{N}
\vphib(a_{\perp},a_{\perp})\vert\psi\rangle
=\sum_{n,m=0}^{N}
\U^{\ast}_{n}\U_{n}\frac{{a^{\ast}_{\phi}}^{2}}{N}\vphib(a_{\perp},a_{\perp})
\U^{\ast}_{m}\U_{m}\vert\psi\rangle 
\\
&=\sum_{n,m=0}^{N}\U^{\ast}_{n}
\Pcal_{n,n} \frac{{a_{\phibar}}^{N-n}}{\sqrt{(N-n)!}} \color{black}
\frac{{a^{\ast}_{\phi}}^{2}}{N}
\frac{{a^{\ast}_{\phi}}^{N-m}}{\sqrt{(N-m)!}}\Pcal_{m,m}
\vphib(a_{\perp},a_{\perp})\U_{m}\vert\psi\rangle
\\
&=\sum_{n=0}^{N-2}b_{N,n}\U^{\ast}_{n}\vphib(a_{\perp},a_{\perp})
\U_{n+2}\vert\psi\rangle~, \quad  b_{N,n}=\frac{\sqrt{(N-n)(N-n+1)}}{N}~,
\end{align*}
since only the terms with $n=m-2$ survive in the last double sum.

{\bf Step 2.} Next, we describe the finite system that results from the above decomposition.
The first observation is that $\U_{n}\vert\psi\rangle=\vert\psi_{N,n}\rangle_{n}$
where $\psi_{N,n}(x_{1},\,\ldots,\, x_{n})$ is a function orthogonal to the condensate. Let 
$\psi_{n}:=\psi_{N,n}$, a function of 
$n$ variables where $n=0,\,1,\,\ldots,\, N$. The operator $h_{\rm ph}$ acts on each of these functions $\psi_{n}$
for $n=1,\,2,\,\ldots,\, N$ by preserving the number of variables.
On the other hand, the operator $\vphib(a_{\perp},a_{\perp})$
maps $\psi_{n+2}$ to $\psi_{n}$. Denote the first action by 
$h_{\rm ph}\circ \psi_{n}$ and the second one 
by $\vphib :\psi_{n+2}$.

We elaborate on these actions. For a symmetric function
$\psi_{n}(x_{1},\,\ldots,\, x_{n})$, we have
\begin{align*}
h_{\rm ph}\circ \psi_{n}
&=d_{n}\sum_{j=1}^{n}
\int \dv y\ \left\{
h_{\rm ph}(x_{j},y)\psi(x_{1},\,\ldots,\, x_{j-1},\,y,\,x_{j+1},\,\ldots,\, x_{n})\right\}~.
\end{align*}
Similarly, $\vphib$ acts on $\psi_{n+2}(x_{1}\ldots x_{n+2})$ 
as follows: 
\begin{align*}
\vphib:\psi_{n+2}&=
b_{n}\int \dv y_{1}\,\dv y_{2}\ 
\left\{\vphib(y_{1},y_{2})
\psi_{n+2}(y_{1},\,y_{2},\,x_{1},\,\ldots,\, x_{n})\right\}~.
\end{align*}
In the above, $d_{n}$ and $b_{n}$ are some (immaterial) numerical constants.

Hence, the eigenvalue equation 
$\Hcal_{\rm ph}\vert\psi\rangle=E\vert\psi\rangle$
reduces to a finite system, viz.,  
\begin{subequations}\label{eqs:fin-sys-eigenv}
\begin{align}
h_{\rm ph}\circ\psi_{N,N}&=E\psi_{N,N}~, \label{eq:fin-sys-eigenv-1}\\
h_{\rm ph}\circ \psi_{N,N-2}+b_{2}\vphib :\psi_{N,N}&=E\psi_{N,N-2}~, \label{eq:fin-sys-eigenv-2} \\
h_{\rm ph}\circ\psi_{N,N-4}+
b_{4}\vphib :\psi_{N,N-2}&=E\psi_{N,N-4}~,\,\ldots\ .\label{eq:fin-sys-eigenv-3}
\end{align}
\end{subequations}
This system has upper triangular form and manifests the effect of pair excitation, since the number of non-condensate particles is reduced in pairs. \color{black} 
The even and odd values of $N$ should be considered separately. These equations describe how to compute the 
fluctuation vector $\vert\psi_{\perp}\rangle=
 \big(\psi_{N,0},\, \psi_{N,1},\, \ldots,\, \psi_{N,N},\, 0,\,\ldots\big)$.
 
Notably, \eqref{eq:fin-sys-eigenv-1} implies the equality of the spectra, 
$\sigma\left(\Hcal_{\rm ph}\right)
 =\sigma\left(h_{\rm ph}(a^{\ast}_{\perp},a_{\perp})\right)$, on $\FF_N$.
Indeed, if~\eqref{eq:fin-sys-eigenv-1}
has only the trivial solution then all the subsequent equations
have trivial solutions. The upper triangular form suggests that 
we can construct the eigenvalues explicitly. Note that the top equation
has infinitely many possible solutions corresponding to the spectrum of $h_{\rm ph}$ -- but choosing one of them results in a finite system of equations.

We now give the relevant construction, which serves as a proof of existence for system~\eqref{eqs:fin-sys-eigenv}. For example, start with (see Lemma~\ref{lem:hphon-spectrum}) \color{black}
\begin{align*}
\Omega_{N}=\prod_{p=1}^{N} \omega_{j_{p}}(x_{p})
\end{align*}
for given $j_p$ ($p=0,\, 1,\, \ldots,\, N$) so that $\Omega_{N}$ is an eigenvector of $h_{\rm ph}(a_\perp^\ast, a_\perp)$, viz., \color{black}
\begin{align*}
h_{\rm ph}\circ\Omega_{N} =\left(\sum_{p=0}^NE_{j_{p}}\right)\Omega_{N}~.
\end{align*}
The action of $\vphib$ on the state $\Omega_N$ \color{black} produces the collection of states
\begin{align*}
\Omega_{l,m}:=\prod_{p=0\atop p\not= l,m}^{N}\omega_{j_{p}}(x_{p})~;\qquad l,\,m=0,\,1,\,\ldots,\, N~.
\end{align*}
We can determine $\Omega_{N-2}:=\sum_{l,m}c_{l,m}\Omega_{l,m}$ which plays the role of $\psi_{N,N-2}$.
By substituting into~\eqref{eq:fin-sys-eigenv-2}, we obtain the system  
$c_{l,m}\big(E_{j_{l}}+E_{j_{m}}\big)=b_{2}
\vphib(\omega_{j_{l}},\omega_{j_{m}})$ which yields $c_{l,m}$. 
The next state, $\Omega_{N-4}$, is a linear combination 
of $\omega_{j_{k}}$ where four terms have been removed from the
original collection. The idea of computation is similar. One can proceed until all  non-condensate particles are removed.
This argument concludes our explicit construction of the eigenvectors of
$\Hcal_{\rm ph}$ in terms of eigenvectors of $h_{\rm ph}$. \color{black} \hfill $\square$

\end{proof}

It is of some interest to observe that the eigenvectors of
$\Hcal_{\rm ph}$ contain (in part) the condensate wave function,
in contrast to the eigenvectors of $h_{\rm ph}(a^{\ast}_{\perp},a_{\perp})$.

%

\section{Connections to a Hermitian approach and $J$-self-adjoint system} 
\label{sec:applications}
In this section, we focus on how  our existence theory for kernel $k$ is connected to  another approach, namely, the use of a (Hermitian) Hamiltonian that does not conserve the number of particles~\cite{fetter72,Lewin2014}. This Hamiltonian results from the Bogoliubov approximation and has the same spectrum as our non-Hermitian $\tHap$ when restricted to $\mathbb{F}_N$. Our analysis reveals a connection between (unitary) Bogoliubov-type rotations of quadratic Hamiltonians, the Riccati equation for $k$, and the theory of $J$-self-adjoint operators developed by Albeverio and coworkers~\cite{AlbeverioMotovilov2019,Albeverio2009,AlbeverioMotovilov2010}  (see also~\cite{Tretter2016,Tretter-book}). These works, however, appear not to address the possible presence of infinitely many solutions to the Riccati equation which is suggested by our existence theory. We also point out that our results so far imply the existence of solutions to the eigenvalue problem for Boson excitations formulated by Fetter, if his delta-function interaction potential is regularized~\cite{fetter72}. \color{black} 

\subsection{On a reduced Hamiltonian via Bogoliubov approximation}
\label{subsec:Bog-eigenvalue}
Recall our reduced Hamiltonian with a smooth interaction potential (Sect.~\ref{subsec:Herm}), viz., 
\begin{equation*}
\Hap=NE_{\rm H}+h(a^{\ast}_{\perp},a_{\perp})
+\frac{1}{2N}f_\phi(a^{\ast}_{\perp},a^{\ast}_{\perp})
a_{\phibar}^{2}
+\frac{1}{2N}\overline{f_\phi}(a_{\perp},a_{\perp}){a^{\ast}_{\phi}}^{2}~.
\end{equation*}
Let us now apply the Bogoliubov approximation to this $\Hap$ by formally replacing the operators $a_{\phibar},\,a^\ast_{\phibar}$ with $\sqrt{N}$. This results in the Hamiltonian $\mathcal{H}_\mathrm{Bog}:\mathbb{F}\mapsto\mathbb{F}$ where 
\begin{equation} \label{eq:Bog-Ham}
    \mathcal{H}_{\mathrm{Bog}} := NE_{\mathrm{H}} 
    + h(a^\ast_\perp,a_\perp)
    +\frac{1}{2}{f_\phi}(a^\ast_\perp,a^\ast_\perp)+\frac{1}{2}\overline{{f_\phi}}(a_\perp,a_\perp)~,
\end{equation}
which does not commute with the number operator $\mathcal N$.

Next, we discuss the diagonalization of $\mathcal{H}_\mathrm{Bog}$ \color{black} by using eigenstates of the operator $h_\mathrm{ph}:\mathfrak{h}^1_V\cap \phi^\perp \to \mathfrak{h}$ defined by~\eqref{eq:hphon} (Sect.~\ref{sec:spectrum}). We proceed in the spirit of Fetter~\cite{fetter72}, who diagonalizes $\mathcal{H}_{\mathrm{Bog}}-NE_H$ \color{black} via (unitary) Bogoliubov-type rotations of the Boson field operators in the space orthogonal to $\phi$; \color{black} see equation (2.14) for a delta-function interaction potential in~\cite{fetter72}. In this vein, let us consider Fetter's  ``quasiparticle'' operators $\gamma_j, \gamma_j^\ast$ which are defined as follows~\cite{fetter72}. 
%
\begin{definition} \label{def:quasi-part}
 The operators $\gamma_j,\, \gamma^\ast_j: \FF\mapsto \FF$ ($j=1,\,2,\,\dots$)  are defined by
\begin{equation*} 
\gamma_j:= \int \dv x\ \{\overline{u_j(x)} a_{\perp, x}+\overline{v_j(x)}a^\ast_{\perp, x}\}~,\quad \gamma_j^\ast:= \int\dv x\ \{u_j(x) a^\ast_{\perp, x}+v_j(x)a_{\perp, x}\}~.
\end{equation*}
In the above, $\{u_j(x)\}_{j=1}^\infty$ is a Riesz basis of $\phi^\perp$, and $\{v_j(x)\}_{j=1}^\infty$ are chosen such that $\gamma_j$ and $\gamma_j^\ast$ satisfy the canonical commutation relations. 
\end{definition}
%

One can verify that $\gamma_j, \gamma^\ast_j$ satisfy the canonical commutation relations provided
\begin{equation}\label{eq:uv-relns}
\begin{split}
\int \dv x\ \{u_j(x)v_{j'}(x) - v_j(x)u_{j'}(x)\} &= 0~,\\   
\int \dv x\ \{u_j(x)\overline{u_{j'}(x)}-v_j(x)\overline{v_{j'}(x)}\} &= \delta_{jj'}~.
\end{split}
\end{equation}
%
We proceed to show that the diagonalization of $\mathcal{H}_{\mathrm{Bog}}$ in terms of $\gamma_j$ and $\gamma_j^*$ implies that $\{u_j(x),v_j(x)\}_{j=1}^\infty$ must solve a linear system of PDEs. Following Fetter's procedure~\cite{fetter72}, let us momentarily assume the following completeness relations:
\begin{equation}\label{eq:complete-relns}
\begin{split}
\sum_{j=1}^\infty{\{u_j(x)\overline{u_j(x')}-\overline{v_j(x)}v_j(x')\}} &= \widehat{\delta}(x,x')~, \\
\sum_{j=1}^\infty{\{u_j(x)\overline{v_j(x')}-\overline{v_j(x)}u_j(x')\}} &=0,\quad\forall x,\,x'\in\mathbb{R}^3~.
\end{split}
\end{equation}
In conjunction with Definition~\ref{def:quasi-part}, these relations allow us to decompose the Boson field operators $a_{\perp, x}$ and $a^\ast_{\perp, x}$ as
\begin{equation*}
a_{\perp, x} = \sum_{j=1}^\infty{\{u_j(x)\gamma_j - \overline{v_j(x)}\gamma_j^\ast\}}~,\quad 
a^\ast_{\perp, x} = \sum_{j=1}^\infty{\{\overline{u_j(x)}\gamma_j^\ast - v_j(x)\gamma_j\}}~.
\end{equation*}
These two relations together with Definition~\ref{def:quasi-part} amount to a Bogoliubov-type (unitary) transformation in the space orthogonal to the condensate $\phi$. 
The substitution of these expressions into~\eqref{eq:Bog-Ham} along with the requirement that the terms proportional to $\gamma_j\gamma_l$ and $\gamma_j^\ast \gamma_l^\ast$ vanish (for all $j,\,l=1,\,2,\,\ldots$) yields the following eigenvalue problem involving a symplectic matrix: 
 \begin{equation}\label{eq:Fetter-PDE}
\begin{pmatrix}
 h^T_\perp & -{f_\phi}_\perp \\
 \overline{{f_\phi}}_\perp & -h_\perp
\end{pmatrix} 
\circ\begin{pmatrix}
u_j(x) \\
v_j(x) 
\end{pmatrix} = E_j
\begin{pmatrix}
u_j(x) \\
v_j(x)
\end{pmatrix}~;\quad j=1,\,2,\,\ldots\,.
\end{equation}
In the above, $h_\perp$ and ${f_\phi}_\perp$ are the projections of operators $h$ and $f_\phi$ on space $\phi^\perp$. System~\eqref{eq:Fetter-PDE} should be compared to equations (2.21a,b) in~\cite{fetter72}. We alert the reader that the notation for $u_j$ and $E_j$ here is the same as the one used for the eigenvectors and eigenvalues of $h_{\rm ph}^\ast$ in Sect.~\ref{subsec:spectrum-hph}. In fact, the corresponding quantities turn out to be identical in the two eigenvalue problems, as we discuss below.  \color{black}

\subsection{On the existence of solutions to eigenvalue problem for $(u_j,\,v_j)$} 
\label{subsec:eigenv-Fet-existence}

Let us recall the spectral theory for $h_{\rm ph}$, particularly Lemma~\ref{lem:hphon-spectrum} (Sect.~\ref{subsec:spectrum-hph}). We should also add that this theory relies on the existence of solutions to the Riccati equation for $k$, Theorem~\ref{thm:Existence} (Sect.~\ref{sec:Riccati}). To make a connection to system~\eqref{eq:Fetter-PDE}, consider the solutions $\omega_j$ and $u_j$ ($j=1,\,2,\,\ldots$) to the eigenvalue problem for $h_{\rm ph}$ and its adjoint. This 
problem is expressed by the equations
\begin{equation*}
\begin{split}
h_\mathrm{ph}(x,\omega_j)= (h+k\circ\overline{f_\phi})(x,\omega_j) &= E_j \omega_j(x)~,\\
h^\ast_\mathrm{ph}(x,u_j) = (h^T+f_\phi\circ\overline{k})(x,u_j) &= E_j u_j(x)\qquad (j=1,\,2,\,\ldots\,)~.
\end{split}
\end{equation*}
Notice that if we define $v_j(x) := -\overline{k}(x,u_j)$ 
then the (adjoint) equation for $u_j(x)$ here immediately takes the form of the first equation in system~\eqref{eq:Fetter-PDE}. We can show that this definition for $v_j$ also gives the second equation in system~\eqref{eq:Fetter-PDE} by employing the conjugate Riccati equation (for $\overline k$). Indeed, notice that 
\begin{equation*}
\begin{split}
-h_\perp (x,v_j) &= (h_\perp \circ\overline{k})(x,u_j)= (-\overline{k}\circ h_\perp^T -\overline{f_\phi}-\overline k\circ f_\phi \circ\overline k)(x,u_j)  \\
&= -\overline{k}\circ \big[E_ju_j(x)-f_\phi\circ\overline{k}(x,u_j)\big] -\overline{f_\phi}(x,u_j)-(\overline k\circ f_\phi \circ\overline k)(x,u_j)  \\
&= -E_j v_j(x) - \overline{f_\phi}(x,u_j)~.
\end{split}
\end{equation*}
Hence, the existence of eigenvectors $\{\omega_j, u_j\}_{j=1}^\infty$ and spectrum $\{E_j\}_{j=1}^\infty$ in regard to $h_{\rm ph}$ entails the existence of solutions to system~\eqref{eq:Fetter-PDE}. 

\begin{remark} 
\textbf{(i)} We showed an intimate connection of the eigenvalue problem for $h_{\rm ph}$ and its adjoint, based on the Riccati equation for kernel $k$, to PDE system~\eqref{eq:Fetter-PDE} coming from Fetter's Hermitian view. A direct comparison to the results in~\cite{fetter72} is meaningful if Fetter's delta-function interaction is appropriately regularized. This connection is in fact a manifestation of a deeper theory which links the Riccati equation to $J$-self-adjoint matrix operators~\cite{AlbeverioMotovilov2019,Albeverio2009,AlbeverioMotovilov2010}. We briefly discuss aspects of this theory in Sect.~\ref{subsec:J-self-adj}. 

\textbf{(ii)} In Fetter's paper~\cite{fetter72}, an ansatz for the many-body ground state, $\vert\psi_0\rangle$, of the quadratic Hamiltonian $\mathcal{H}_\mathrm{Bog}$ on $\FF$ is 
\begin{equation*}
\vert\psi_0\rangle = Ze^{\mathcal G} \{{a_{\phibar}^\ast}^N\}|vac\rangle~;\quad \mathcal G := \frac{1}{2}\int \dv x\,\dv y\ b(x,y)\,a^\ast_{\perp, x}a^\ast_{\perp, y}~.
\end{equation*}
However, a single governing equation for the associated kernel $b(x,y)$ is not given in~\cite{fetter72}. Instead, the condition $\gamma_j |\psi_0\rangle =0$ is applied for all $j=1,2,\dots$, which yields the following system of integral relations:
\begin{equation*}
\int \dv y\ \big\{b(x,y)\overline{u}_j(y)\big\} = -\overline{v_j(x)}\qquad (j=1,\,2,\,\ldots)~. 
\end{equation*}
Evidently, by comparison of this formalism to our approach,  we realize that kernel $b$ coincides with $k$, and the above integral relations are already a consequence of our solution for $k$. In fact, in~\cite{fetter72} Fetter uses the above integral system to define the kernel 
$b(x,y)=k(x,y)$ when $\{u_j(x),v_j(x)\}$ solve the matrix eigenvalue problem \eqref{eq:Fetter-PDE} (under a delta-function interaction). Our existence proof for $k$ furnished in the context of Theorem~\ref{thm:Existence} shows that the ground state $\vert\psi_0\rangle$ is self-consistent, in the sense that the integral system stemming from $\vert\psi_0\rangle$ and \eqref{eq:Fetter-PDE} is well-posed if a solution to the Riccati equation for $k$ exists. 
\end{remark}

At this stage, we find it compelling to give the following corollary for system~\eqref{eq:Fetter-PDE}.
%
\begin{corollary} \label{cor:Cor1}
For  $\{u_j(x), v_j(x)\}_{j=1}^\infty$ that solve \eqref{eq:Fetter-PDE}, completeness relations~\eqref{eq:complete-relns} and orthogonality relations~\eqref{eq:uv-relns} hold.
%
%
\end{corollary}
%

\begin{proof}
We resort to the spectral theory of operator $h_{\rm ph}$ on space $\phi^\perp$, particularly the proof of Lemma~\ref{lem:hphon-spectrum} (Sect.~\ref{subsec:spectrum-hph}). 
Recall the completeness relation for the basis $\{\eta_j\}_{j=1}^\infty$ of $\phi^\perp$, as well as the relation $\eta_j(x) = (\widehat\delta-k\circ\overline k)^{1/2}(x,u_j)$. 

Hence, on $\phi^\perp$ we have
\begin{equation*}
\begin{split}
\widehat{\delta}(x,x') &= (\widehat\delta - k\circ\overline k)^{1/2}\left\{\sum_{j=1}^\infty{\eta_j(x)\overline{\eta_j(x')}}\right\}(\widehat\delta - k\circ\overline k)^{-1/2} \\
&= (\widehat\delta - k\circ\overline k)\left\{\sum_{j=1}^\infty{u_j(x)\overline{\eta_j(x')}}\right\}(\widehat\delta - k\circ\overline k)^{-1/2} = (\widehat\delta-k\circ\overline k)\sum_{j=1}^\infty{u_j(x)\overline{u_j(x')}}~.
\end{split}
\end{equation*}
Thus, we obtain    
\begin{equation*}
\sum_{j=1}^\infty{u_j(x)\overline{u_j(x')}}=\frac{\widehat\delta(x,x')}{\widehat\delta - k\circ\overline{k}}~.
\end{equation*}
By use of the relation $v_j(x) = -\overline{k}(x,u_j)$, we can therefore assert that
\begin{equation*}
\sum_{j=1}^\infty{\overline{v_j(x)}v_j(x')} = \sum_{j=1}^\infty k(\overline{u_j},x) {\overline k(x',u_j)}= \frac{k\circ \overline{k}}{\widehat\delta- k\circ \overline k}~.
\end{equation*}
The last two equations entail the first relation of~\eqref{eq:complete-relns}.

Next, we invoke the equation just derived to write 
\begin{equation*}
\sum_{j=1}^\infty{u_j(x)\overline{v_j(x')}} = -\sum_{j=1}^\infty u_j(x) k(x',\overline{ u_j})= -(\widehat{\delta}-k\circ \overline{k})^{-1}\circ k~.
\end{equation*}
Alternatively, we have
\begin{equation*}
\sum_{j=1}^\infty{\overline{v_j(x)}u_j(x')} = -\sum_{j=1}^\infty{k(x,\overline{u_j})u_j(x')} = -k\circ (\widehat{\delta}- \overline k\circ k)^{-1}~.
\end{equation*}
Thus, we obtain the second completeness relation of~\eqref{eq:complete-relns} by using the identity $(\widehat{\delta}-k\circ \overline{k})^{-1}\circ k=k\circ (\widehat{\delta}- \overline k\circ k)^{-1}$.

Regarding orthogonality relations \eqref{eq:uv-relns}, the manipulation of system~\eqref{eq:Fetter-PDE} yields the following equations:
\begin{equation*}
\begin{split}
(E_j - \overline{E_{j'}}) \int \dv x\ \{u_j(x)\overline{u_{j'}(x)}\}&=\int \dv x\ \{-\overline{u_{j'}(x)}f_\phi(x,v_j)+u_j(x)\overline{f_\phi}(x,\overline{v_{j'}})\}~,\\
(E_j - \overline{E_{j'}}) \int \dv x\ \{v_j(x)\overline{v_{j'}(x)}\}&=\int \dv x\ \{\overline{v_{j'}(x)}\,\overline{f_\phi}(x,u_j)-v_j(x)f_\phi(x,\overline{u_{j'}})\}~.
\end{split}
\end{equation*} 
By subtracting the second equation from the first one, we can obtain the second orthogonality relation of~\eqref{eq:uv-relns}, if $\Vert u_j\Vert^2_2- \Vert v_{j}\Vert^2_2 \neq 0$ and this normalization for $u_j$ and $v_j$ is chosen to give unity. The first orthogonality relation of~\eqref{eq:uv-relns} follows by a similar procedure which we omit here.  \color{black}  \hfill $\square$
\end{proof}

\subsection{On the $J$-self-adjoint system}
\label{subsec:J-self-adj}
Next, we discuss in more detail the connection between Riccati equation~\eqref{eq:k-Riccati} and main ideas from the theory of $J$-self-adjoint operators found in, e.g.,~\cite{AlbeverioMotovilov2019,Albeverio2009,AlbeverioMotovilov2010}. A link between these two theories is suggested by the eigenvalue problem~\eqref{eq:Fetter-PDE}, which involves the symplectic matrix
\begin{equation} \label{eq:M-def}
M:=\begin{pmatrix}
     h^T_\perp & -{f_\phi}_\perp \\
     \overline{f_\phi}_\perp & -h_\perp
    \end{pmatrix};\quad\mathrm{dom}(M):=\mathfrak{h}^1_V\oplus\mathfrak{h}^1_V~.
\end{equation}
Note the matrix 
\begin{equation*}
\widetilde{M}:=\begin{pmatrix}
     h^T & -{f_\phi} \\
     \overline{f_\phi} & -h
    \end{pmatrix}
\end{equation*}
has the zero eigenvalue with eigenvector $(\phi, \phibar)$.
\color{black}

Suppose that $\phi(x)$, $h(x,y)$ and ${f_\phi}(x,y)$ satisfy the assumptions of Theorem~\ref{thm:Existence} (Sect.~\ref{subsec:exist-thm-proof}). Let $k$ be the unique solution to the Riccati equation with $\|k\|_{\mathrm{op}}<1$. Then the operator matrix
\begin{equation*}
W := \begin{pmatrix}
\widehat\delta & \quad k \\
 \overline{k} & \quad \widehat\delta 
\end{pmatrix} : \phi^\perp\oplus \phi^\perp \mapsto \phi^\perp\oplus\phi^\perp
\end{equation*}
is boundedly invertible~\cite{Albeverio2009}, with inverse 
\begin{equation*}
W^{-1} = \begin{pmatrix}
 (\widehat\delta-k\circ \overline k)^{-1} & -k\circ(\widehat\delta-\overline k\circ k)^{-1} \\
 -\overline{k}\circ(\widehat\delta-k\circ \overline{k})^{-1} & (\widehat\delta-\overline{k}\circ k)^{-1}
\color{black}\end{pmatrix}~. 
\end{equation*}

Now let us consider the diagonal matrix 
\begin{equation*} 
D := \begin{pmatrix}
 h^T_\perp + k\circ\overline{f_\phi}_\perp & 0 \\
 0 & -h_\perp-\overline{k}\circ{f_\phi}_\perp
\end{pmatrix}~.
\end{equation*}
The spectrum of $D$ is $\sigma(\overline{h_\mathrm{ph}})\cup\sigma(-h_\mathrm{ph})$, which under the assumptions of Theorem~\ref{thm:Existence} consists of two disjoint parts. 
Since $k$ obeys the Riccati equation on $\phi^\perp$, we have 
\begin{equation*}
DW = \begin{pmatrix}
 h^T_\perp + k\circ\overline{f_\phi}_\perp & \quad -k\circ h_\perp - {f_\phi}_\perp \\
 \overline{k}\circ h^T_\perp +\overline{f_\phi}_\perp & \quad -h_\perp-\overline{k}\circ{f_\phi}_\perp
\end{pmatrix} = WM~,
\end{equation*}
where $M$ is defined by~\eqref{eq:M-def}. Thus, the matrix $M$ is similar to the diagonal matrix $D$.

We proceed to describe implications of this similarity relation. Eigenvectors of the diagonal operator matrix $D$ are of two types. One type is of the form $ \big({\omega}_j(x),0\big)$ where ${\omega}_j(x)$ is an eigenvector of ${h_\mathrm{ph}}$, and another type is of the form $\big(0,\overline{\omega_j(x)}\big)$ (see Sect.~\ref{subsec:spectrum-hph}). This fact  yields two types of eigenvectors for $M$ after transformation by $W^{-1}$, viz.,
\begin{equation*}
W^{-1}\begin{pmatrix}
 {\omega}_j(x) \\
 0
\end{pmatrix} = \begin{pmatrix}
 (\widehat\delta-k\circ \overline{k})^{-1}(x,{\omega}_j) \\
 -\overline{k}\circ(\widehat\delta-k\circ \overline{k})^{-1}(x,\omega_j)
\end{pmatrix}~,
\end{equation*}
and
\begin{equation*}
W^{-1}\begin{pmatrix}
0\\
\overline{\omega_j(x)}
\end{pmatrix} = \begin{pmatrix}
-k\circ(\widehat\delta-\overline{k}\circ k)^{-1}(x,\overline{\omega_j})\\
(\widehat\delta-\overline{k}\circ k)^{-1}(x,\overline{\omega_j})
\end{pmatrix}~.
\end{equation*}
For the second type of eigenvector, we make the identifications 
\begin{equation*}
u_j(x):=(\widehat\delta-k\circ \overline{k})^{-1}(x,\omega_j)\quad \mbox{and}\quad 
v_j(x):= -\overline{k}\circ(\widehat\delta-k\circ \overline{k})^{-1}(x, \omega_j)= -\overline{k}(x,u_j)~.
\end{equation*}
The eigenvectors of the second type should be excluded because they yield a negative spectrum.   
%
\begin{remark}
So far, we assumed that the Riccati equation for $k$ is satisfied (and solutions to this equation exist by Theorem~\ref{thm:Existence}). \color{black} Conversely, if we assume that the integral system $v_j=-\overline{k}(x, u_j)$ as well as PDE system~\eqref{eq:Fetter-PDE} hold then $k$ must obey the Riccati equation. This claim can be proved by use of the methods that we already developed.   	
\end{remark}

\section{Conclusion}
\label{sec:conclusion}
In concluding this paper, we stress the intimate, and perhaps surprising, mathematical connection between two apparently disparate approaches (those of Fetter~\cite{fetter72} and Wu~\cite{wu61}) to the problem of Boson excitations via the theory of $J$-self-adjoint operators~\cite{Albeverio2009}. The results presented here form an application of its powerful machinery to a physics-inspired problem with interesting implications. 

Notably, the similarity relation $WMW^{-1} = D$, discussed in Sect.~\ref{subsec:J-self-adj}, shows that the spectrum of $h_\mathrm{ph}$ can change for different solutions to the Riccati equation, but in a predictable way. In particular, since the spectrum  $\sigma(M) = \sigma(D) = \sigma(h_\mathrm{ph})\cup\sigma(-h_\mathrm{ph})$, the (double) spectrum $\sigma(h_\mathrm{ph})\cup\sigma(-h_\mathrm{ph})$ is unaffected by the choice of $k$ solving the Riccati equation. However, the spectrum $\sigma(h_\mathrm{ph})$ will change under different choices of solutions for $k$. In light of our analysis, the only possible change induced by $\sigma(h_\mathrm{ph}(k))\mapsto\sigma(h_\mathrm{ph}(k'))$ for two different solutions $k$ and $k'$ ($k'\not=k$) is such that a finite collection  of eigenvalues $\{E_j\}\subset \sigma(h_\mathrm{phon}(k))$ is mapped to $\{-E_j\}\subset\sigma(h_\mathrm{phon}(k'))$ while the rest of the eigenvalues remain unchanged. 

We are tempted to mention a few open problems motivated by our work. For example, given the existence of the kernel $k$ with $\Vert k\Vert_{\text{op}}<1$, it is of interest to study the existence of the Boson pair correlation function in a trap at zero temperature. Another possible extension is to consider the effect of a non-unitary transformation analogous to $e^{\mathcal W}$ by including contributions from higher-order (cubic and quartic terms) in the reduced many-body Hermitian Hamiltonian. This consideration would plausibly require the introduction of additional kernels, which must satisfy several consistency conditions. Finally, it is conceivable that the non-Hermitian approach involving $k$ can be extended to the setting of finite (positive) temperatures below the phase transition in the presence of a trapping potential. In the spirit of the periodic case~\cite{leeyang}, we could construct an effective quadratic Hamiltonian that involves a parameter expressing the average fraction of particles at the condensate, and subsequently transform it non-unitarily. Alternatively, one may use a Hermitian approach at finite temperatures akin to Fetter's formalism, e.g., the approach of~\cite{Griffin1996}.

\color{black}

\begin{acknowledgements}
The second author (D.M.) is indebted to Professor Tai Tsun Wu for valuable discussions on the Bose-Einstein condensation. The authors are grateful to Dr.~Eite Tiesinga for bringing Ref.~\cite{fetter72} to their attention and Professor Matei Machedon for various discussions on Boson dynamics. 
\end{acknowledgements}





\bibliographystyle{spmpsci}      

\bibliography{GMS_ArXiv-Revision}

\end{document}